%% file: main-efron-stein.tex
\def\confversion{0}
\def\ifconf{\ifnum\confversion=1}
\def\ifnotconf{\ifnum\confversion=0}

\documentclass[11 pt]{article}

\def\showauthornotes{1}
\def\showkeys{0}
\def\showdraftbox{0}

\input{macros}

\usepackage[linesnumbered, ruled,vlined]{algorithm2e}
\usepackage{tikz,pgf}

\usepackage{caption}
\newcommand{\Poincare}{Poincar\'{e}\xspace}

\newcommand{\mat}[1]{\mathbf{#1}}
\newcommand{\HH}{\mathbb{H}}

\newcommand{\resamp}[1]{\widetilde{#1}}
\newcommand{\coef}[2]{\widehat{#1}(#2)}
\newcommand{\dpoly}{d_p}
\newcommand{\icL}{\cL^{-1}}

\newcommand{\RR}{\mathbb{R}}
\newcommand{\NN}{\mathbb{N}}

\newcommand{\EE}{\mathbb{E}}

\newcommand{\sig}{\sigma}
\newcommand{\Sig}{\Sigma}
\newcommand{\al}{\alpha}

\newcommand{\Gam}{\Gamma}
\newcommand{\gam}{\gamma}
\newcommand{\Del}{\Delta}

\newcommand{\grad}{\nabla}

\newcommand{\graphmat}[1]{\mat{M}_{#1}}
\newcommand{\sch}[2]{\norm{#1}_{#2}^{#2}}
\newcommand{\Etr}[1]{\EE \tr\left[#1\right]}
\newcommand{\Esch}[2]{\EE \norm{#1}_{#2}^{#2}}
\newcommand{\var}[1]{\Varsymb[#1]}
\newcommand{\herm}[1]{\overline{#1}}

\newcommand{\T}{\intercal}

\newtheorem{propn}[theorem]{Proposition}

\newcommand{\mA}{{\mat{A}}}
\newcommand{\mB}{{\mat{B}}}
\newcommand{\mC}{{\mat{C}}}
\newcommand{\mD}{{\mat{D}}}

\newcommand{\mF}{{\mat{F}}}
\newcommand{\mG}{{\mat{G}}}
\newcommand{\mH}{{\mat{H}}}
\newcommand{\mI}{{\mat{I}}}

\newcommand{\mK}{{\mat{K}}}

\newcommand{\mM}{{\mat{M}}}
\newcommand{\mN}{{\mat{N}}}

\newcommand{\mR}{{\mat{R}}}

\newcommand{\mU}{{\mat{U}}}
\newcommand{\mV}{{\mat{V}}}

\newcommand{\mX}{{\mat{X}}}

\newcommand{\mDel}{{\mat{\Del}}}
\newcommand{\mPi}{{\mat{\Pi}}}
\newcommand{\mSig}{{\mat{\Sig}}}

\newcommand{\cG}{{\cal{G}}}

\newcommand{\cI}{{\cal{I}}}
\newcommand{\cJ}{{\cal{J}}}
\newcommand{\cK}{{\cal{K}}}
\newcommand{\cL}{{\cal{L}}}

\newcommand{\cP}{{\cal{P}}}

\newcommand{\cR}{{\cal{R}}}
\newcommand{\cS}{{\cal{S}}}

\newcommand{\iid}{i.i.d.\xspace}

 \usepackage[top=1in, bottom=1in, left=1in, right=1in]{geometry}

\usepackage{blkarray}
\usepackage{bigstrut}
\begin{document}

\title{Concentration of polynomial random matrices via Efron-Stein inequalities}
\author{
Goutham Rajendran\thanks{University of Chicago. Supported in part by NSF grants CCF-1816372 and CCF-2008920.  {\tt goutham@uchicago.edu}. }
\and
Madhur Tulsiani \thanks{Toyota Technological Institute at Chicago. Supported by NSF grant CCF-1816372. \tt madhurt@ttic.edu}
}

\setcounter{page}{0}
\date{}
\maketitle
\thispagestyle{empty}

 \vspace{-15 pt}

\abstract{Analyzing concentration of large random matrices is a common task in a wide variety of fields. Given independent random variables, several tools are available to bound the norms of random matrices whose entries are linear in the variables, such as the matrix-Bernstein inequality.
However, for many recent applications, we need to bound the norms of random matrices whose entries are polynomials in the variables.
Such matrices arise naturally in the analysis of spectral algorithms (e.g., Hopkins et al. [STOC 2016], Moitra and Wein [STOC 2019]), and in lower bounds for semidefinite programs based on the Sum-of-Squares (SoS) hierarchy (e.g. Barak et al. [FOCS 2016], Jones et al. [FOCS 2021]).


\medskip

In this work, we present a general framework to obtain such bounds, based on the beautiful matrix Efron-Stein inequalities developed by Paulin, Mackey and Tropp [Annals of Probability 2016]. The Efron-Stein inequality bounds the norm of a random matrix by the norm of another potentially simpler (but still random) matrix. We view the latter matrix as arising by ``differentiating'' the starting matrix.
By recursively differentiating, our framework reduces the main task to bounding the norms of far simpler matrices.
These simpler matrices are in fact deterministic matrices in the case of Rademacher random variables and hence, bounding their norm is a far easier task.
In general for non-Rademacher random variables, the task reduces to the much easier task of scalar concentration. Moreover, in the setting of polynomial matrices, our main result also generalizes the work of Paulin, Mackey and Tropp.
%


\medskip

As applications of our basic framework, we recover known bounds in the literature, especially for simple ``tensor networks'' and ``dense graph matrices''. As applications of our general framework, we derive bounds for ``sparse graph matrices''. The sparse graph matrix bounds were obtained only recently by Jones et al.  [FOCS 2021] using a nontrivial application of the trace power method, and was a core component in their work.
We expect this framework will also be helpful for other applications involving concentration phenomena for nonlinear random matrices.
}

\newpage

\ifnotconf
\pagenumbering{roman}
\tableofcontents
\clearpage
\fi

\pagenumbering{arabic}
\setcounter{page}{1}


\section{Introduction}\label{sec: intro}
\input{intro}

\section{Preliminaries}\label{sec: prelims}
\input{prelims}

\section{The basic framework for Rademacher random variables} \label{sec: basic_recursion}

\input{rademacher_recursion}

\section{Applications}\label{sec:rademacher-applications}

To illustrate our framework, we apply it to obtain concentration bounds for nonlinear random matrices that have been considered in the literature before. The first application is a simple tensor network that arose in the analysis of spectral algorithms for a variant of principal components analysis (PCA) \cite{hopkins2015tensor, hopkins2018statistical}.
The second application is to obtain norm bounds on dense graph matrices \cite{medarametla2016bounds, ahn2016graph}. In the second application, the norm bounds are governed by a combinatorial structure called \textit{the minimum vertex separator of a shape}. We will show how this notion arises naturally under our framework, whereas prior works that derived such bounds used the trace power method and required nontrivial combinatorial insights.

\subsection{A simple tensor network}

\input{tensor_network_norm_bound}

\subsection{Graph matrices}\label{sec: dense_graph_matrices}

\input{dense_graph_matrices}

\section{Why a na\"ive application of \cite{paulin2016} may fail for general product distributions} \label{sec: failure_of_basic}

\input{failure_of_basic}

\section{The general recursion framework}\label{sec: general_recursion}

\input{general_recursion}

\section{A generalization of \cite{paulin2016} and proof of \cref{lem: main_general}}\label{sec: proof_of_general}

\input{proof_of_general}

\section{Application: Sparse graph matrices} \label{sec: sparse_graph_matrices}

\input{sparse_graph_matrices}

\section*{Acknowledgements}
We thank Aaron Potechin, Chris Jones and Jeff Xu for helpful discussions regarding graph matrices.
We thank Pierre Youssef for helpful pointers to references.
We thank anonymous reviewers for their suggestions and thank Chris Jones and Aaron Potechin for proofreading an earlier version of this manuscript.

\bibliographystyle{alpha}
\bibliography{macros,madhur}

\end{document}

%% file: macros.tex
\usepackage{xspace,xcolor,enumerate}
\usepackage{amsmath,amssymb}
\usepackage{amsthm}
\usepackage[toc,page]{appendix}
\usepackage{thmtools}
\usepackage{thm-restate}
\usepackage{color,graphicx}
\usepackage{boxedminipage}

\ifnum\showkeys=1
\usepackage[color]{showkeys}
\fi

\definecolor{darkred}{rgb}{0.5,0,0}
\definecolor{darkgreen}{rgb}{0,0.35,0}
\definecolor{darkblue}{rgb}{0,0,0.55}

\usepackage[pdfstartview=FitH,pdfpagemode=UseNone,colorlinks,linkcolor=darkblue,filecolor=darkred,citecolor=darkgreen,urlcolor=darkred,pagebackref]{hyperref}

\usepackage[capitalise,nameinlink]{cleveref}
\usepackage[T1]{fontenc}
\usepackage{mathtools,dsfont,bbm}
\usepackage{mathpazo}

\ifnum\showauthornotes=1
\newcommand{\Authornote}[2]{{\sf\small\color{red}{[#1: #2]}}}
\newcommand{\Authorcomment}[2]{{\sf \small\color{gray}{[#1: #2]}}}
\newcommand{\Authorfnote}[2]{\footnote{\color{red}{#1: #2}}}
\else
\newcommand{\Authornote}[2]{}
\newcommand{\Authorcomment}[2]{}
\newcommand{\Authorfnote}[2]{}
\fi

\ifnum\showdraftbox=1

\else

\fi

\newtheorem{theorem}{Theorem}[section]

\newtheorem{definition}[theorem]{Definition}
\newtheorem{lemma}[theorem]{Lemma}
\newtheorem{remark}[theorem]{Remark}

\newtheorem{corollary}[theorem]{Corollary}
\newtheorem{claim}[theorem]{Claim}
\newtheorem{fact}[theorem]{Fact}

\newtheorem{remk}[theorem]{Remark}
\newtheorem{example}[theorem]{Example}

\def\FullBox{\hbox{\vrule width 6pt height 6pt depth 0pt}}

\def\qed{\ifmmode\qquad\FullBox\else{\unskip\nobreak\hfil
\penalty50\hskip1em\null\nobreak\hfil\FullBox
\parfillskip=0pt\finalhyphendemerits=0\endgraf}\fi}

\def\qedsketch{\ifmmode\Box\else{\unskip\nobreak\hfil
\penalty50\hskip1em\null\nobreak\hfil$\Box$
\parfillskip=0pt\finalhyphendemerits=0\endgraf}\fi}

\def\to{\rightarrow}
\def\epsilon{\varepsilon}

\def\eps{\epsilon}

\def\phi{\varphi}
\def\cal{\mathcal}

\newcommand{\ie}{i.e.,\xspace}

\newcommand{\etal}{et al.\xspace}

\newcommand{\mper}{\,.}
\newcommand{\mcom}{\,,}

\newcommand{\R}{{\mathbb R}}

\newcommand{\B}{\{0,1\}\xspace}
\newcommand{\pmone}{\{-1,1\}\xspace}

\usepackage{nicefrac}

\newcommand{\abs}[1]{\ensuremath{\left\lvert #1 \right\rvert}}

\newcommand{\norm}[1]{\ensuremath{\left\lVert #1 \right\rVert}}

\newcommand{\Esymb}{\mathbb{E}}
\newcommand{\Psymb}{\mathbb{P}}

\newcommand{\Varsymb}{\mathrm{Var}}
\DeclareMathOperator*{\ExpOp}{\Esymb}

\makeatletter
\def\Pr#1{%
    \ProbabilityRender{\Psymb}{#1}%
}

\def\Ex#1{%
    \ProbabilityRender{\Esymb}{#1}%
}

\def\Var#1{%
    \ProbabilityRender{\Varsymb}{#1}%
}

\def\ProbabilityRender#1#2{
  \@ifnextchar\bgroup%
  {\renderwithdist{#1}{#2}}
   {\singlervrender{#1}{#2}}
}
\def\singlervrender#1#2{%
   \ensuremath{\mathchoice
       {{#1}\left[ #2 \right]}
       {{#1}[ #2 ]}
       {{#1}[ #2 ]}
       {{#1}[ #2 ]}
   }
}
\def\renderwithdist#1#2#3{%
   \@ifnextchar\bgroup
   {\superfancyrender{#1}{#2}{#3}}
   {\ensuremath{\mathchoice
      {\underset{#2}{#1}\left[ #3 \right]}
      {{#1}_{#2}[ #3 ]}
      {{#1}_{#2}[ #3 ]}
      {{#1}_{#2}[ #3 ]}
     }
   }
}
\def\superfancyrender#1#2#3#4#5{
   \ensuremath{\mathchoice
      {\underset{#1}{{#1}}\left#4 #3 \right#5}
      {{#1}_{#2}#4 #3 #5}
      {{#1}_{#2}#4 #3 #5}
      {{#1}_{#2}#4 #3 #5}
   }
}
\makeatother

\newfont{\inhead}{eufm10 scaled\magstep1}

\newcommand{\polylog}{{\mathrm{polylog}}}

\DeclareMathOperator\supp{Supp}

\DeclareMathOperator{\tr}{\operatorname {tr}}

\newcommand{\bigoh}{\operatorname{O}}

\newcommand{\inparen}[1]{\left(#1\right)}             
\newcommand{\insquare}[1]{\left[#1\right]}             

\newcommand{\Erdos}{Erd\H{o}s\xspace}
\newcommand{\Renyi}{R\'enyi\xspace}



%% file: intro.tex
In  optimization, statistics, and spectral algorithms, we often want to understand the concentration of various random matrices. To do this, we can appeal to the powerful theory of matrix-deviation inequalities~\cite{tropp2015:book}.
For example, the matrix-Bernstein inequality addresses random matrices of the form
\[
\mat{M} ~=~ x_1 \cdot \mat{C_1} + \cdots + x_n \cdot \mat{C_n}
\]
where $x_1, \ldots, x_n$ are independent scalar random variables, and $\mat{C_1}, \ldots, \mat{C_n}$ are fixed matrices.
A large selection of such inequalities are available when the random matrix (say) $\mat{M}$ is a \emph{linear} function of independent random variables. However, several recent works require us to understand random matrices which are \textit{non-linear} functions, and in particular low-degree polynomial functions, of scalar random variables. This forms the focus of our work.

As a motivating example, consider the random matrix $\mat{M} \in \R^{[n]^2 \times [n]^2}$ obtained as
\[
\mat{M} ~=~ \mat{A_1} \otimes \mat{A_1} + \cdots + \mat{A_m} \otimes \mat{A_m} \mcom
\]
where $\mat{A_1}, \ldots, \mat{A_m} \in \R^{[n]\times[n]}$ are independent random matrices, with \iid entries uniformly distributed in $\pmone$.
It is easy to see that the entries of the matrix $\mat{M}$ are degree-2 polynomial functions of the independent random variables describing the entries of $\mat{A_1}, \ldots, \mat{A_m}$. The concentration of such a matrix was analyzed by Hopkins \etal \cite{hopkins2015tensor, hopkins2018statistical}, who use it to design spectral algorithms for a variant of the principal components analysis (PCA). This matrix is a special case of a more general setting that we study in this work.

\paragraph{Matrix-valued polynomial functions.}
In the example above, the entries of the matrices are low-degree polynomials in independent (Rademacher) random variables.
In this work, we consider a general setting where we take an $n$-tuple $Z = (Z_1, \ldots, Z_n)$ of independent and identically distributed random variables\footnote{Our framework also applies when the variables are not necessarily identically distributed, as long as they are independent.} distributed in $\Omega$.
We consider random matrices given by a matrix-valued function $\mat{F}(Z)$ taking values in $\R^{\cI \times \cJ}$ for arbitrary index sets $\cI, \cJ$, where each entry $\mat{F}[I,J](Z)$ is a polynomial in $Z_1, \ldots, Z_n$.
We develop a general framework to analyze concentration of such matrices.
Our matrix concentration results are simpler to state in the case when $Z_1, \ldots, Z_n$ are independent Rademacher variables uniformly distributed in $\pmone$, but apply for the general case as well.

Special cases of such non-linear random matrices have been used in several applications in spectral
algorithms and lower bounds. We now briefly discuss a few examples below.
Note that while the previous methods used for these examples have been somewhat problem-specific,
the goal of this work is to develop a general method. While our techniques \emph{also} apply for these
examples (providing a proof of concept), understanding these examples is not required to follow our results.
A reader only interested in the techniques for obtaining concentration, may also choose to skip
ahead to the next section directly.

\begin{enumerate}
\item \textbf{Tensor networks.}
Random matrices such as the above were viewed as a special case of ``flattened tensor networks'' by Moitra and Wein~\cite{moitra2019spectral}, who also considered spectral algorithms obtained via somewhat larger tensor networks.
A tensor network is a graph with nodes corresponding to tensors (see the figure below for an example). An edge between two nodes corresponds to shared indices for one of the dimensions and the degree of each node is equal to the order of the corresponding tensor (the number of dimensions).
%
Such networks indicate how tensors of different orders can be multiplied to obtain larger ones.
%
For example, the first network in the figure below illustrates the network corresponding to simple multiplication $\mA \cdot \mB$ of two matrices $\mat{A} \in \R^{m \times n}$ and $\mat{B} \in \R^{n \times m}$, where the red and blue edges indicate the row and column indices respectively.
Similarly, the second network in the figure below illustrates the network corresponding to the application by Hopkins \etal \cite{hopkins2016fast}, where $\mat{T} \in \R^{n  \times n \times m}$ is a random tensor with \iid entries in $\pmone$.
%
While the latter network yields an order-4 tensor, they obtain a matrix in $\RR^{n^2 \times n ^2}$ by ``flattening'' it, where the row is indicated by the indices in the red edges and the column is indicated by the indices in the blue edges.
%
In the figure, we also indicate the index sets corresponding to each of the edges (though these are often supressed in the diagrams).
Moitra and Wein~\cite{moitra2019spectral} analyzed a larger tensor network, with a graph consisting of 10 nodes, in their algorithm for the continuous multi-reference alignment problem.
\begin{figure}[h]
\label{fig:tensor-network-example}
\begin{center}
\begin{tikzpicture}
\draw (1,2) node {{\large $\mat{A}$}};
\draw (3,2) node {{\large $\mat{B}$}};
\draw (2,2.25) node {\small {$[n]$}};
\draw (0,2.25) node {\small {$[m]$}};
\draw (4,2.25) node {\small {$[m]$}};
\draw (1.25, 2) -- (2.75,2);
\draw[red]  (0.75, 2) -- (-0.5,2);
\draw[blue] (3.25, 2) -- (4.5,2);
\draw (8,2) node {{\large $\mat{T}$}};
\draw (10,2) node {{\large $\mat{T}$}};
\draw (8.25, 2) -- (9.75,2);
\draw[red] (7.75, 1.75) -- (6.5,1);
\draw[blue] (7.75, 2.25) -- (6.5,3);
\draw[red] (10.25, 1.75) -- (11.5,1);
\draw[blue] (10.25, 2.25) -- (11.5,3);
\draw (9,2.25) node {\small {$[m]$}};
\draw (6.6,1.5) node {\small {$[n]$}};
\draw (6.6, 2.5) node {\small {$[n]$}};
\draw (11.4,1.5) node {\small {$[n]$}};
\draw (11.4, 2.5) node {\small {$[n]$}};
\end{tikzpicture}
\end{center}
%
\caption{Tensor networks for matrix multiplication and the algorithm in \cite{hopkins2016fast}}
\end{figure}
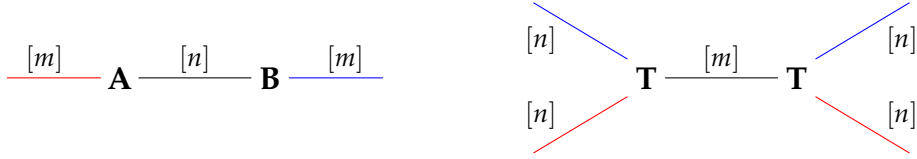

\item \textbf{Graph matrices.}
Another setting of nonlinear concentration arises from the analysis of the so-called ``graph matrices'' ~\cite{medarametla2016bounds, ahn2016graph}. Graph matrices play an important role in lower bounds for average-case problems, against algorithms based on the powerful Sum-of-Squares (SoS) SDP hierarchy running in polynomial time and even sub-exponential time~\cite{meka2015sum, deshpande2015improved, hopkins2015sos, raghavendra2015tight, barak2019nearly, mohanty2020lifting, ghosh2020sum, sparsesos, rajendran2022combinatorial, potechin2022subexponential, Jones:thesis}.

Let $\mat{X}$ be the $\{\pm1\}$-adjacency matrix of a random graph in $\cG_{n,1/2}$ \ie $\mat{X}[i,j]$ is uniform $\pmone$ when $i \neq j$ and 0 when $i=j$.
Graph matrices are random matrices corresponding to the occurences of a small graph pattern called a ``shape''.
A shape $\tau$ is a small, fixed graph with two ordered subsets $U_{\tau}, V_{\tau}$ of vertices. For simplicity, let $\tau$ be a shape of a fixed size, where the vertex set $V(\tau)$ is partitioned into two ordered sets $V(\tau) = U_{\tau} \sqcup V_{\tau}$.
%
For such a shape $\tau$, the corresponding \emph{graph matrix} $\mat{M}_{\tau}$ has rows and columns indexed by $[n]^{|U_{\tau}|}$ and $[n]^{|V_{\tau}|}$ respectively, and we view the row and column indices $I$ and $J$ as defining a (unique in this case) map $\phi: U_{\tau} \sqcup V_{\tau} \to [n]$. The corresponding entry is given by
\[
\mat{M}_{\tau}[I,J]
~=~ \mat{M}_{\tau}[\phi(U_{\tau}),\phi(V_{\tau})]
~=~
\begin{cases}
\prod_{(u,v)\in E(\tau)}\mat{X}[\phi(u), \phi[v]] & \text{if}~\phi~\text{is injective} \\[5 pt]
0 & \text{otherwise}
\end{cases}
\]
In the case of general graph matrices (defined formally in \cref{sec: dense_graph_matrices}), $U_{\tau}, V_{\tau}$ are arbitrary ordered subsets of the vertex set of $\tau$, and we sum over all feasible injective maps $\phi$.
As an example, consider the case shown in \cref{fig: tau}, where $\tau$ is a triangle on three vertices $\{u_1, v_1, v_2\}$ with $U_{\tau} = (u_1)$ and $V_{\tau} = (v_1, v_2)$. Then, the corresponding matrix is given by
\[
\mat{M}_{\tau}[i_1,(i_2,i_3)] ~=~ \mat{X}[i_1,i_2] \cdot \mat{X}[i_2,i_3] \cdot \mat{X}[i_3,i_1] \mcom
\]
where $\mat{X}$ automatically enforces injectivity.
%

Graph matrices are closely related to tensor networks (ignoring the injectivity constraint on $\phi$). For instance, the above matrix can be viewed as the flattened tensor network below, where the tensor $\mat{I}$ denotes the ``diagonal'' tensor of order 3 with entries being 1 if all indices are equal and 0 otherwise.
%
\begin{figure}[h]
\begin{center}
\begin{tikzpicture}
\draw (1,4) node {$u_1$};
\draw (1,4) circle (0.5 cm);
\draw[dashed, blue] (0.35,3.35) rectangle (1.65, 4.65);
\draw (3,5) node {$v_1$};
\draw (3,5) circle (0.5 cm);
\draw (3,3) node {$v_2$};
\draw (3,3) circle (0.5 cm);
\draw[dashed, red] (2.35,2.35) rectangle (3.65, 5.65);
\draw (1.5,4) -- (2.5,5);
\draw (1.5,4) -- (2.5,3);
\draw (3,4.5) -- (3,3.5);
\draw (8,4) node {\large $\mat{I}$};
\draw (11,5.5) node {\large $\mat{I}$};
\draw (11,2.5) node {\large $\mat{I}$};
\draw (9.5,4.75) node {\large $\mat{X}$};
\draw (11,4) node {\large $\mat{X}$};
\draw (9.5,3.25) node {\large $\mat{X}$};
\draw (8.25,4.15) -- (9.25,4.65);
\draw (10.75,5.35) -- (9.75,4.85);
\draw (8.25,3.85) -- (9.25,3.4);
\draw (10.75,2.65) -- (9.75,3.15);
\draw (11,2.75) -- (11,3.75);
\draw (11,5.25) -- (11,4.25);
\draw[blue] (7.75,4) -- (7,4);
\draw[red] (11.15,5.65) -- (11.7,6.2);
\draw[red] (11.15,2.35) -- (11.7,1.8);
\end{tikzpicture}
\end{center}
\caption{The graph $\tau$ and corresponding flattened tensor network}
\label{fig: tau}
\end{figure}
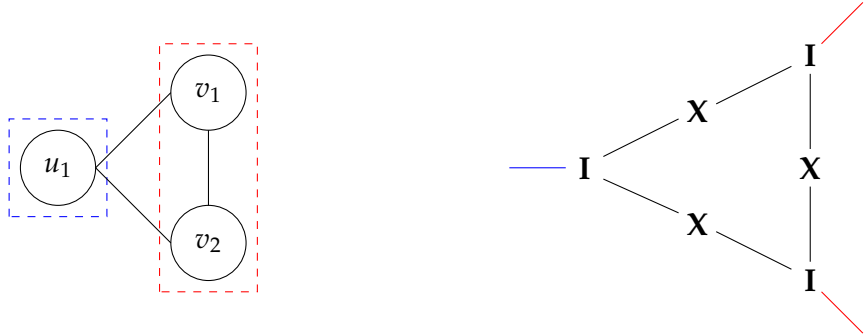

\end{enumerate}

\paragraph{Analyzing concentration}
Recall that our objective is to analyze the concentration of polynomial random matrices.
To motivate our approach, consider first the problem of obtaining concentration bounds on a \emph{scalar} polynomial $f(Z)$ with mean zero. To obtain such bounds, because of Markov's inequality, it suffices to compute moment estimates
\[
\Pr{\abs{f(Z)} \geq \lambda}
~=~ \Pr{\inparen{f(Z)}^{2t} \geq \lambda^{2t}}
~\leq~ \lambda^{-2t} \cdot {\Ex{\inparen{f(Z)}^{2t}}}
\]
While in some cases $\Ex{\inparen{f(Z)}^{2t}}$ can be computed by direct expansion, it often involves an intricate analysis of the structure of terms with degrees growing with $t$, and therefore indirect methods may be more convenient.
One such method is based on hypercontractive inequalities.
In particular for Rademacher variables, the hypercontractive inequality~\cite{ODonnell08} gives that for a polynomial $f$ of degree $d_p$, we have
\[
\Ex{\inparen{f(Z)}^{2t}} ~\leq~ (2t-1)^{d_p \cdot t} \cdot \inparen{\Ex{\inparen{f(Z)}^2}}^t \mper
\]
Thus, for (scalar) polynomial functions, the hypercontractive inequality gives moment estimates using $\inparen{f(Z)}^2$, which is convenient because $\inparen{f(Z)}^2$ is a polynomial of \emph{fixed} degree and therefore is much easier to understand. In fact, it can often be conveniently analyzed using the Fourier coefficients of $f$.

The matrix analog of the above argument involves the Schatten-$2t$ norm $\norm{.}_{2t}$, which is defined for a matrix $\mat{M} $ with non-zero singular values $\sigma_1, \ldots, \sigma_r$ as $\norm{\mat{M}}_{2t}^{2t} ~:=~ \sum_{j \in [r]} \sigma_j^{2t}$.
For a function $\mat{F}$ with $\Ex{\mat{F}(Z)} = 0$, we have the following bound using Schatten norms.
\[
\Pr{\sigma_1(\mat{F}) \geq \lambda}
~\leq~ \lambda^{-2t} \cdot \Esch{\mF}{2t} 
~=~ \lambda^{-2t} \cdot \Etr{\inparen{\mF(Z)\mF(Z)^\T}^{t}}
\]
Known norm bounds for tensor networks~\cite{moitra2019spectral} (which involves Gaussian variables)
and graph matrices~\cite{ahn2016graph, sparsesos}, start with the above inequality, and rely on
direct expansion of the trace.
They analyze terms in the expansion as being formed by $2t$ copies of the network/shape, with
possibly overlapping vertex sets. To analyze such graphs, they both rely on intricate combinatorics,
as well as arguments relying crucially on the problem structure.
%


In terms of general techniques, while hypercontractive inequalities are also known for matrix-valued functions of Rademacher variables~\cite{BARDW08}, their form involves Schatten-$p$ norms for $p \in [1,2]$ and (to the best of our knowledge) are not known to imply matrix concentration.
%
To get around this, we consider another indirect method based on Efron-Stein inequalities. In the scalar case, Efron-Stein inequalities give us a slight weakening of the above scalar bound. Interestingly, it turns out that this can indeed be generalized to the matrix case.

\paragraph{Efron-Stein inequalities.}
Efron-Stein inequalities bound the global variance of a function of independent random variables, in terms of local variance estimates obtained by changing one variable at a time.
For $i \in [n]$ and tuple $Z = (Z_1, \ldots, Z_n)$, let $Z^{(i)}$ denote the tuple $(Z_1, \ldots, Z_{i-1}, \resamp{Z_i}, Z_{i+1}, \ldots, Z_n)$, where $\resamp{Z_i}$ is an independent copy of $Z_i$.
For a scalar function $f(Z)$, the Efron-Stein inequality states that
\[
\Var{f(Z)} ~=~ \Ex{\inparen{f(Z) - \EE f}^2} ~\leq~ \frac12 \cdot \sum_{i \in [n]} \Ex{\inparen{f(Z) - f\inparen{Z^{(i)}}}^2} ~=~ \Ex{V(Z)}\mcom
\]
where $V(Z) ~:=~ \sum_{i \in [n]} \EE \insquare{\inparen{f(Z) - f\inparen{Z^{(i)}}}^2 | Z}$.
For Rademacher variables, $\Ex{V(Z)}$ is equal to the total influence from boolean Fourier analysis and indeed, the above inequality can also be observed via Fourier analysis. In fact, when $f$ is a polynomial of degree $d_p$, the two sides are within a factor $d_p$.

A moment version of the Efron-Stein inequality was developed by Boucheron \etal~\cite{BBLM05}, who obtain bounds in terms of $V(Z)$ (in fact, in terms of more refined quantities $V_+(Z)$ and $V_{-}(Z)$) which serves as a proxy for the variance. Their results imply that for a function $f$,
\[
\Ex{\inparen{f(Z) - \EE f}^{2t}} ~\le~ (C_0 \cdot t)^t \cdot \Ex{\inparen{V(Z)}^{t}} \mper
\]
A beautiful matrix generalization of the above inequality (\cref{thm: main_efron_stein} below) was obtained by Paulin, Mackey and Tropp~\cite{paulin2016}, via the method of exchangeable pairs (see also~\cite{HT21:poincare} for a different proof).
Their inequality is stated for Hermitian matrix valued functions $\mH$. But we can also use it for non-Hermitian functions $\mat{F}$, where we simply apply it to the Hermitian dilation $\mH = \begin{bmatrix}
	0 & \mF\\
	\mF^\T & 0
\end{bmatrix}$ instead.
\begin{theorem}[\cite{paulin2016}]\label{thm: main_efron_stein}
Let $\mat{H}(Z)$  be a Hermitian matrix valued function of independent random variables $Z = (Z_1, \ldots, Z_n)$ with $\ExpOp\norm{\mat{H}} < \infty$.
Then, for each natural number $t \ge 1$,
\[
\Etr{(\mat{H} - \EE\mat{H})^{2t}}  ~\le~  (4t-2)^t \cdot \Etr{\mV^t} \mcom
\]
where $\mV(Z)$ is the variance proxy defined as
\[
\mV(Z) := \frac{1}{2} \cdot \sum_{i = 1}^n \Ex{\inparen{\mat{H}(Z) - \mat{H}\inparen{Z^{(i)}}}^2 \mid Z} \mper
\]
\end{theorem}

\paragraph{A simple bound for Rademacher variables.}
The form of the variance proxy suggests a recursive approach for polynomial functions (say of degree $d_p$) of Rademacher variables. Consider the scalar case again, where we assume without loss of generality that $f$ is multi-linear. In particular, consider the Efron-Stein inequality by Boucheron \etal~\cite{BBLM05}, where the variance proxy can be written as
\begin{align*}
V(Z)
~=~ \frac12 \cdot \sum_{i \in [n]} \Ex{\inparen{f(Z) - f\inparen{Z^{(i)}}}^2 \mid Z}
&~=~ \frac12 \cdot \sum_{i \in [n]} \Ex{(Z_i - \resamp{Z_i})^2 \cdot \inparen{\frac{\partial f(Z)}{\partial Z_i}}^2 \mid Z} \\
&~=~ \sum_{i \in [n]} \inparen{\frac{\partial f(Z)}{\partial Z_i}}^2
~=~ \norm{\mat{f}_1(Z)}_2^2 \mcom
\end{align*}
where $\mat{f}_1(Z)$ is a vector-valued function given by $\mat{f}_1[i](Z) = \frac{\partial f(Z)}{\partial Z_i}$. Thus, to estimate $\ExpOp \inparen{f(Z)}^{2t}$, we just need to estimate $\ExpOp \norm{\mat{f}_1(Z)}_2^{2t}$, where $\mat{f}_1(Z)$ is now a vector valued function. The key observation is that $\mat{f}_1(Z)$ has entries of degree at most $d_p - 1$. This suggests that we can apply this inequality recursively until we end up with constant polynomials, which we fully understand.
We can do a similar computation for matrix-valued functions $\mat{F}(Z)$ using \cref{thm: main_efron_stein}. This yields two matrices $\mat{F}_{0,1}$ and $\mat{F}_{1,0}$ of partial derivatives, where an extra index $i$ is added either to the row or column indices. Iterating this yields the following result, which we state in terms of the partial derivative operators
$\grad_{\alpha}(f) = \inparen{\prod_{i: \alpha_i = 1} \frac{\partial}{\partial Z_i}}(f)$ for $\alpha \in \B^n$ (extended entry-wise to matrices).
\begin{restatable}[Rademacher recursion]{theorem}{basicframework}\label{thm: main_rademacher}
Let $\mat{F}: \pmone^n \to \R^{\cI \times \cJ}$ be a matrix valued polynomial function of degree at most $d_p$. Then, for each natural number $t \ge 1$,
\[
\Esch{\mF - \EE \mF}{2t} ~\le~ \sum_{1 \le a + b \leq \dpoly} (16t\dpoly)^{(a + b) \cdot t} \cdot \sch{\EE\mF_{a, b}}{2t} \mcom
\]
where $\mF_{a,b}$ is a matrix of partial derivatives indexed by the sets $\cI \times \binom{[n]}{a}$ and $\cJ \times \binom{[n]}{b}$ with
\[
\mF_{a, b}[(\cdot, \al), (\cdot, \beta)] = \begin{dcases}
	\grad_{\alpha + \beta}(\mF) & \text{ if $\al\cdot \beta = 0$}\\
	0 & \text{otherwise}
      \end{dcases}
    \]
where $\alpha, \beta \in \B^n$ are indicator vectors of sets in $\binom{[n]}{a}$ and
$\binom{[n]}{b}$ respectively.
\end{restatable}
\begin{remark}
While we state our results in terms of polynomial moment bounds, it is also possible to obtain
exponential tails using these results. This can be done either using an appropriate (known) variant
of \cref{thm: main_efron_stein}, or by using a sufficiently large value of $t$.
These results can also recover (known) matrix Chernoff or Bernstein inequalities when the function $\mF$ is
linear, but of course the much more interesting case is when $\mF$ is a polynomial function.
\end{remark}
%
%

We cover some applications of the above theorem in \cref{sec:rademacher-applications}.
Similar to the hypercontractive bound for the scalar case, the bound above is in terms of a small
number ($O(d_p^2)$) of matrices that arise from polynomials of fixed degree (not growing with $t$),
but importantly, they are \emph{deterministic} matrices. Because they are deterministic, analyzing
them is considerably easier.
Note that bounds depending on norms of a fixed number of deterministic matrices, arise even in the
study of concentration for \emph{scalar} polynomial functions~\cite{AW15}, and thus it is not
surprising that they are needed to control the much more challenging case of matrix-valued functions.

When we apply this theorem to the case $\mF = \mM_{\tau}$, the graph matrix of a shape $\tau$, we obtain bounds in terms of combinatorial objects known as ``vertex separators'' of the shape $\tau$. This recovers the bounds by Ahn \etal~\cite{ahn2016graph} and perhaps surprisingly (to the authors), this gives an alternative and direct derivation of these combinatorial structures such as vertex separators, compared to the ingenious observations made in Ahn \etal~\cite{ahn2016graph}.
The important takeaway is that these norm bounds can be recovered by our more general technique rather than relying on problem-specific methods.

\paragraph{Extending the framework to general product distributions.}
A key contribution of our work is to show how the above framework can be extended to arbitrary product distributions (with bounded moments).
A motivating example of this is norm bounds for the so-called ``sparse graph matrices''. In sparse graph matrices, the variables $Z_i$ can be thought of as (normalized) edges of a $\cG_{n', p}$ graph, that is, $Z_i = -\sqrt{\frac{1-p}{p}}$ with probability $p$ and $Z_i = \sqrt{\frac{p}{1-p}}$ with probability $1-p$. These variables are standard in $p$-biased Fourier analysis~\cite{o2014analysis} and are chosen to satisfy $\ExpOp Z_i = 0$ and $\ExpOp Z_i^2 = 1$. Sparse graph matrices naturally arise when analyzing average case problems on $\cG_{n, p}$ graphs for $p = o(1)$, as opposed to $\cG_{n, 1/2}$ graphs.

Until recently, little was known about norm bounds for sparse graph matrices. The difficulty stems partly from the fact that when $p = o(1)$, it is important that sparse graph matrix norm bounds have the right dependence on $p$ and not just on $n$. Such norm bounds were obtained recently by Jones \etal~\cite{sparsesos}, via the trace power method which involved a delicate combinatorial counting argument.
On the other hand, we obtain similar norm bounds using our framework but in a more mechanical fashion.
%
%
We can also readily apply our framework in the even more general case of sub-Gaussian random variables and our bounds will depend on the sub-Gaussian norm of the distributions.

To extend our framework to general product distributions, we could take inspiration from the Rademacher case and could attempt to simply recursively apply the Efron-Stein inequality. Unfortunately, this idea will fail. The issue can be observed by again considering the scalar case.
%
Assume that $Z_1, \ldots, Z_n$ are \iid with $\ExpOp Z_i = 0$ and $\ExpOp Z_i^2 = 1$ for all $i \in [n]$.
Also assume for simplicity that $f(Z)$ is a multi-linear polynomial of degree $\dpoly$. Analyzing the variance proxy as before, we get
\[
V(Z)
~=~ \frac12 \cdot \sum_{i \in [n]} \Ex{(Z_i - \resamp{Z_i})^2 \cdot \inparen{\frac{\partial f(Z)}{\partial Z_i}}^2 \mid Z}
~=~ \frac12 \sum_{i \in [n]} \Ex{(Z_i - \resamp{Z_i})^2 | Z} \cdot \inparen{\frac{\partial f(Z)}{\partial Z_i}}^2 \mper
\]
In the Rademacher case, we had $\Ex{(Z_i - \resamp{Z_i})^2 | Z} = 2$. This left us with the polynomials corresponding to partial derivatives but which importantly had a strictly lower degree. However, for a general product distribution, we instead have $\Ex{(Z_i - \resamp{Z_i})^2 | Z} = 1+Z_i^2$. This gives back a term $\inparen{Z_i \cdot \frac{\partial f}{\partial Z_i}}^2$ where the polynomial inside the square could have degree possibly still equal to $d_p$.
%
This means that in the next step of the recursion, we may again have to consider a derivative with respect to $Z_i$ and may again end up with the same polynomial $f$. Therefore, the recursion is stalled! A similar issue occurs for matrices, which is elaborated in \cref{sec: failure_of_basic}. To get around this, we generalize the work of \cite{paulin2016}.

\paragraph{Generalizing \cite{paulin2016} via explicit inner kernels.}
To resolve the above issue, we modify the proof of \cite{paulin2016} and our proof techniques may be of independent interest.

We first recall how the matrix Efron-Stein inequality, \cref{thm: main_efron_stein}, was proved in \cite{paulin2016}. Their basic strategy is to utilize the theory of \textit{exchangeable pairs} \cite{stein1972bound, stein1986approximate, chatterjee2005concentration, chatterjee2006stein}, in particular \textit{kernel Stein pairs}.
%
A kernel Stein pair is an exchangeable pair of random matrices that has a ``kernel'', a bivariate function that ``reproduces'' the matrices in the pair.
More concretely, consider an exchangeable pair of random variables $(Z,Z')$ (which means $(Z',Z)$ has the same distribution). For this exchangeable pair, a bivariate matrix-valued function $\mK(z,z')$ is said to be a kernel for a matrix-valued function $\mF$ if it satisfies
\begin{itemize}
    \item Anti-symmetry: $\mK(z',z) = -\mK(z,z')$ for all inputs $(z,z')$.
    \item Reproducing property: $\Ex{\mK(Z,Z') \mid Z} = \mF(Z)$.
\end{itemize}
If such a kernel $\mK$ exists, then the pair of random variables $(\mF(Z), \mF(Z'))$ is said to be a kernel Stein pair.

Building on ideas from \cite{stein1986approximate, chatterjee2005concentration}, Paulin, Mackey and Tropp~\cite{paulin2016} first show the existence of a kernel, by exhibiting it as a limit of coupled Markov Chains. By studying the evolution of this kernel coupling, they prove analytic properties of the kernel.
Then, using this kernel, they employ the powerful method of exchangeable pairs to evaluate moments of the random matrix, which in turn will imply concentration.

For a Hermitian random matrix $\mX$, they introduce two matrices - the \textit{conditional variance} $\mV_{\mX}$ which measures the squared fluctuations of $\mX$ when resampling a coordinate of $Z$; and the \textit{kernel conditional variance} $\mV^{\mK}$ which measures the squared fluctation of the kernel when resampling a coordinate of $Z$. With these matrices in hand, they bound the Schatten $2t$-norm of $\mX$ by the Schatten $t$-norm of $s\mV_{\mX} + s^{-1}\mV^{\mK}$ for any parameter $s > 0$. Finally, they choose $s$ appropriately to make these two quantities approximately equal, in which case it simplifies to the variance proxy $\mV$, proving \cref{thm: main_efron_stein}.

In our setting, no such choice of $s$ is feasible because for any choice of $s$, either the conditional variance term $s\mV_{\mX}$ will dominate $\mX^2$ or the kernel conditional variance term $s^{-1}\mV^{\mK}$ will dominate $\mX^2$. This will make the main inequality \cref{thm: main_efron_stein} trivial.

To get around this, we will exploit the structure of the matrix we have, i.e. $\mF = \mD\mG\mD$ where $\mD$ is a diagonal matrix that encodes all variables that have already been differentiated on and $\mG$ is a polynomial matrix of the remaining variables. Since $\mD$ is a simple diagonal matrix with low degrees, most of the deviations exhibited by $\mF$ are in fact likely to be exhibited by $\mG$. To capture this intuition, we consider a kernel for only the inner matrix $\mG$ instead of $\mF$ as a whole. We call this an \textit{inner kernel}.

This helps us avoid the root cause of the issue, i.e. differentiating on variables we have already encountered (which correspond to entries in $\mD$).
Therefore, the recursion will not stall!

However, in general, this is not realizable since $\mD$ and the kernel of $\mG$ can interact in unexpected ways. To study this interaction, we construct explicit polynomial kernels (\cref{thm: explicit_kernel_for_poly}) (compared to \cite{paulin2016} who show the existence of the kernel but for all functions).

We study how this explicit inner kernel interacts with $\mD$ (see \cref{lem: props_of_exp_kernel_mat}) and use it to obtain a generalization of the inequalities by \cite{paulin2016} (generalized because setting $\mD = \mI$ will give back their result) stated in \cref{lem: main_pmt_bound}.

A subtle issue is that the conditional variance of $\mX$ may still have additional deviations due to the diagonal matrices $\mD$ (which still involve random variables). We control the additional deviations using Jensen's operator trace inequality (for non-commuting averages)~\cite{hansen2003jensen} (stated in \cref{lem: jensen_trace}).
Putting these ideas together lets us obtain a version of the Efron-Stein inequality where the variance proxy only corresponds to the conditional variance of the inner kernel. In the setting of polynomial functions, this inequality generalizes the work of \cite{paulin2016}.

With the modified Efron-Stein inequality from above, we cannot guarantee that the matrices $\mF$ at intermediate steps are of lower degree, but on the other hand, the degree of the inner matrix $\mG$ reduces at each step. Therefore, we can recursively apply this inequality to obtain our final bounds. The final bounds are then stated in terms of norm bounds for the simplified matrices of the form $\mD\mG\mD$ where $\mG$ are deterministic matrices and $\mD$ are diagonal matrices which are still functions of $Z$.
%
While random, these matrices can be easily analyzed via simple scalar concentration tools.
%

The main theorem is stated in \cref{sec: general_recursion}, in particular \cref{thm: main_general}, with the proof following in \cref{sec: proof_of_general}. While our proof builds on the work by \cite{paulin2016}, the argument here is self-contained.

\paragraph{Applications.}
Our framework is suitable for many nonlinear concentration results obtained in the literature \cite{BarakBHKSZ12, ge2015decomposing, hopkins2015tensor, medarametla2016bounds, ahn2016graph, hopkins2016fast, schramm2017fast, hopkins2018statistical, hopkins2019robust, moitra2019spectral, kivva2020exact, potechin2020machinery, sparsesos, bafna2022polynomial, rajendran2022nonlinear, Jones:thesis}.
%
We show a few of these applications in \cref{sec:rademacher-applications} and \cref{sec: sparse_graph_matrices}.
We expect similar future applications to benefit from our framework because the task is mechanically reduced to analyzing considerably simpler matrices.

In \cref{sec: dense_graph_matrices}, we derive norm bounds on dense graph matrices. In earlier works, dense graph matrices have been used extensively in analysis of semidefinite programming hierarchies, especially the Sum-of-Squares (SoS) hierarchy \cite{meka2015sum, deshpande2015improved, hopkins2015sos, raghavendra2015tight, barak2019nearly, mohanty2020lifting, ghosh2020sum, potechin2022subexponential, rajendran2022combinatorial}. For more applications and a detailed treatment of graph matrices, see \cite{ahn2016graph, Jones:thesis}.

In \cref{sec: sparse_graph_matrices}, we derive norm bounds for sparse graph matrices. Sparse graph matrices have been relatively less understood until recently, when \cite{sparsesos} obtained norm bounds for such matrices via the trace power method. They use these bounds to prove SoS lower bounds for the maximum independent set problem on sparse graphs.

\paragraph{Other related work}\label{par: related}
Nonlinear concentration for the case of scalar-valued functions has been the subject of an extensive
body of work. In addition to the results of Schudy and Sviridenko~\cite{schudy2011bernstein} which
we use, strong concentration results have also been obtained (for example) in the results of
Lata{\l}a~\cite{Latala06}, Adamczak and Wolff~\cite{AW15}, and Bobkov, G{\"o}tze, and
Sambale~\cite{BGS19}. In addition to the above results, hypercontractive inequalities can also be
used to obtain concentration inequalities for low-degree (scalar) polynomial
functions~\cite{ODonnell08} with possible sub-optimal exponents.

For the case of matrix-valued functions, while we rely here on the work of Paulin, Mackey, and
Tropp~\cite{paulin2016} for product distributions, later works have also extended these results to
distributions satisfying weaker assumptions. In particular, Aoun, Banna, and Youssef~\cite{ABY20}
obtained matrix concentration for distributions satisfying matrix Poincar{\'e} inequalities,
building on earlier work of Cheng, Hsieh, and Tomamichel~\cite{CHT17, CH19}.
It was later proved by Garg, Kathuria, and Srivastava~\cite{GKS21} that the matrix Poincar{\'e}
inequalities are implied by scalar Poincar{\'e} inequalities.
Independently, matrix concentration based on scalar Poincar{\'e} inequalities was also proved by
Huang and Tropp~\cite{HT21:poincare}. Another work of Huang and Tropp~\cite{HT21:semigroup} also
extablishes matrix concentration inequalities via semigroup methods. While some hypercontractive
inequalities are also known for matrix-valued functions~\cite{BARDW08, AD21:hypercontractivity}, to the best of our knowledge, they do not
imply concentration bounds for matrices with low-degree polynomial entries.

\paragraph{Potential extensions}\label{par: extension}

In this work, we assumed that the input forms a product distribution. In other words, the variables $Z_1, \ldots, Z_n$ are independent. A natural extension is the case when they are not independent. This has important applications for many problems such as when the input is a uniform $d$-regular graph, or when the input is sampled from a distribution with a global constraint, etc. In such cases, the input variables are not independent but it may be possible to use similar ideas to analyze concentration.

More concretely, to study concentration in the non-independent setting, one can use the recent work of Huang and Tropp~\cite{HT21:poincare} on matrix concentration from \Poincare inequalities, together with our framework. For this, we just need to exhibit a Markov process that converges to our desired distribution.

\paragraph{Organization of the paper and bibliographic note.}

We start with preliminaries in \cref{sec: prelims}. In \cref{sec: basic_recursion}, we state and prove the Rademacher recursion. We illustrate some applications of this framework in \cref{sec:rademacher-applications}. In \cref{sec: failure_of_basic}, we explain why similar ideas may not be enough in the general case. We then propose our general framework in \cref{sec: general_recursion} and prove it in \cref{sec: proof_of_general}. We end with an application of the general framework to sparse graph matrices in \cref{sec: sparse_graph_matrices}. An earlier version of this paper also appeared in the dissertation of the first author \cite{rajendran2022nonlinear}.



%% file: prelims.tex
\paragraph{Notation}

We use boldface letters such as $\mI, \mM, \mX\ldots, $ to denote matrices.
Entries of a matrix $\mX \in \R^{\cI \times \cJ}$ will be denoted by $\mX[I,J]$ for $I \in \cI, J \in \cJ$. Let $\HH^n$ denote the set of $n \times n$ real symmetric matrices. The trace of a matrix $\mX \in \HH^n$ equals $\sum_{i \in [n]} \mX[i,i]$ and is denoted by $\tr \mX$.

\subsubsection*{Multi-index notation}

For any pair of vectors $\al, \beta \in \NN^n$ and scalar $c \in \NN$, we define $\al + \beta, \al \cdot \beta, c\al$ entrywise. We also define the orderings $\al \le \beta$ and $\al \unlhd \beta$ where we say $\al \le \beta$ if for each $i$, $\al_i \le \beta_i$, and $\al \unlhd \beta$ if for each $i$, $\al_i$ is either $0$ or $\beta_i$. We denote by $|\al|_0$ the number of nonzero entries of $\al$ and by $|\al|_1$, the sum of entries of $\al$. For a boolean vector $\gam \in \{0, 1\}^n$, we define $1 - \gam$ the vector with all its bits flipped.




\subsubsection*{Derivatives}

For variables $Z_1, \ldots, Z_n$ and $\al \in \NN^n$, define the monomial $Z^{\al} := \prod_{i = 1}^n Z_i^{\al_i}$. This forms a standard basis for polynomials.

For $\al \in \NN^n$, we define the linear operator $\grad_{\al}$ that acts on polynomials by defining its action on the elements $Z^{\beta}$ as follows and then extend linearly to all polynomials.
\[\grad_{\al}(Z^{\beta}) = \begin{dcases}
	Z^{\beta - \al} & \text{ if $\al \unlhd \beta$}\\
	0 & \text{ o.w.}
\end{dcases}\]

Informally, for a polynomial $f$ written as a linear combination of the standard basis polynomials $Z^{\beta}$, $\grad_{\al}(f)$ isolates the terms that precisely contain the powers $Z_i^{\al_i}$ for all $i$ such that $\al_i \neq 0$ and then truncates these powers. In other words, it's the coefficient of $Z^{\alpha}$ in $f$. In particular, observe that $\grad_{\al}(f)$ does not depend on $Z_i$ for any $i$ such that $\al_i \neq 0$.

Supose $f$ is multilinear, as we can assume in the Rademacher case when we are working with $Z_i \in \pmone$. For $\al \in \{0, 1\}^n$ with nonzero indices $i_1, \ldots, i_k \in [n]$, we have $\grad_{\al}(f) = \frac{\partial}{\partial Z_{i_1}}\ldots \frac{\partial}{\partial Z_{i_k}}f$. So this linear operator generalizes the partial derivative operator. But note that in general, $\grad$ is not simply the standard partial derivative operator.


\subsubsection*{Matrix Analysis}

Linear operators that act on polynomials can also be naturally defined to act on matrices by acting on each entry.

We define $\mI_m$ to be the $m \times m$ identity matrix. We drop the subscript when it's clear.
For matrices $\mF, \mG$, define $\mF \oplus \mG$ to be the matrix $\begin{bmatrix}
	0 & \mF\\
	\mG & 0
\end{bmatrix}$. For a matrix $\mF$, define its Hermitian dilation $\herm{\mF}$ as $\mF \oplus \mF^T$. Denote by $\preceq$ the Loewner order, that is, $\mA \preceq \mB$ for $\mA, \mB \in \HH^n$ if and only if $\mB - \mA$ is positive semi-definite.

\begin{definition}
	For a matrix $\mF$ and an integer $t \ge 0$, define the Schatten $2t$-norm as
	\[\norm{\mF}_{2t}^{2t} = \tr[{(\mF\mF^T)^t}]\]
\end{definition}

\begin{fact}\label{fact: cs}
	For real symmetric matrices $\mX_1, \ldots, \mX_n$, we have
	\begin{align*}
		(\mX_1 + \ldots + \mX_n)^2 \preceq n(\mX_1^2 + \ldots + \mX_n^2)
	\end{align*}
\end{fact}

\begin{fact}\label{fact: holder}
	For positive semidefinite matrices $\mX, \mX_1, \ldots, \mX_n$ such that $\mX \preceq \mX_1 + \ldots + \mX_n$ and for any integer $t \ge 1$,
	\begin{align*}
		\tr [\mX^t] \le n^{t - 1}(\tr[\mX_1^t] + \ldots + \tr[\mX_n^t])
	\end{align*}
\end{fact}
\begin{proof}
    By H\"{o}lder's inequality, $n^{t - 1}(\tr[\mX_1^t] + \ldots + \tr[\mX_n^t]) \ge (\norm{\mX_1}_t + \ldots + \norm{\mX_n}_t)^t$. By triangle inequality of Schatten norms, this is at least $\norm{\mX_1 + \ldots + \mX_n}_t^t$. Finally, because $\mX_1 + \ldots + \mX_n \succeq \mX\succeq 0$, we can use the monotonicity of trace functions (see \cite[Proposition 1]{petz1994survey}) where we use the increasing function $f(x) = x^t$ on $x \in [0, \infty)$. This proves the result.
\end{proof}

\begin{lemma}[Jensen's operator trace inequality]\cite[Corollary 2.5]{hansen2003jensen}\label{lem: jensen_trace}
	Let $f$ be a convex, continuous function defined on an interval $I$ and suppose that $0 \in I$ and $f(0) \le 0$. Then, for all integers $m, n \ge 1$, for every tuple $\mB_1, \ldots, \mB_n$ of real symmetric $m \times m$ matrices with spectra contained in $I$ and every tuple $\mA_1, \ldots, \mA_n$ of $m \times m$ matrices with $\sum_{i = 1}^n \mA_i^T\mA_i \preceq \mI$, we have
	\[\tr[f(\sum_{i = 1}^n \mA_i^T \mB_i \mA_i)] \le \tr[\sum_{i = 1}^n \mA_i^T f(\mB_i) \mA_i]\]
\end{lemma}
%


%% file: rademacher_recursion.tex
%
%
Let $Z = (Z_1, \ldots, Z_n)$ be sampled uniformly from $\{-1, 1\}^n$.
We will consider matrix-valued functions $\mF: \pmone^n \to \RR^{\cI \times \cJ}$, with rows and columns indexed by arbitrary sets $\cI, \cJ$ respectively such that for all $I \in \cI, J \in \cJ$,
\[
\mF[I, J] ~=~ f_{I, J}(Z)
\]
where $f_{I, J}$ are polynomials of $Z_1, \ldots, Z_n$.
Since $Z_i \in \{-1, 1\}$, we can assume without loss of generality that $f_{I, J}$ are multilinear.
Let $\dpoly$ be the maximum degree of any $f_{I, J}$ in $\mF$.
In this section, we will give a general framework using which we can obtain bounds on
$\Esch{\mF - \EE\mF}{2t}$ for any integer $t \ge 1$. We restate the theorem for convenience.

%







\basicframework*


\begin{remark}
	To obtain high probability norm bounds from moment estimates, we can set $t = \polylog(n)$ and invoke Markov's inequality. Since we do not attempt to optimize the dependence on the logarithmic factors, we do not attempt to optimize the exponent of $t$ in the main theorem.
\end{remark}

To prove this, we will prove \cref{lem: main_rademacher} and then recursively apply it.
For each $i \le n$, define the random vector
\[Z^{(i)} ~:=~ (Z_1, \ldots, Z_{i - 1}, \resamp{Z_i}, Z_{i + 1}, \ldots, Z_n)\]
where $\resamp{Z_i}$ is an independent copy of $Z_i$, that is,
is independently resampled from $\{-1, 1\}$.

Let $\mX := \mF- \EE\mF$. When the input is $Z$, we denote the matrices as $\mF, \mX$, etc and when the input is $Z^{(i)}$, denote the corresponding matrices as $\mF^{(i)}, \mX^{(i)}$, etc. That is, for $I \in \cI, J \in \cJ$, we have $\mF^{(i)}[I, J] = f_{I, J}(Z^{(i)})$. Define $\mX_{a, b} = \mF_{a, b} - \EE\mF_{a, b}$.

\begin{lemma}\label{lem: main_rademacher}
	For integers $a, b \ge 0$, we have
	\[\Esch{\mX_{a, b}}{2t} \le (16t\dpoly)^t(\Esch{\mX_{a, b + 1}}{2t} + \Esch{\mX_{a + 1, b}}{2t} + \sch{\EE\mF_{a, b + 1}}{2t} +\sch{\EE\mF_{a + 1, b}}{2t})\]
\end{lemma}

Using this lemma, we can complete the proof of the main theorem.

\begin{proof}[Proof of \cref{thm: main_rademacher}]
	Observing that $\mX$ is a principal submatrix of $\mX_{0, 0}$ with all other entries being $0$, we can apply \cref{lem: main_rademacher} repeatedly until $\mX_{a, b} = 0$, which will be the case if $a + b > \dpoly$.
\end{proof}

In the rest of this section, we will prove \cref{lem: main_rademacher}. We start with a basic fact. Let $\mat{e}_i \in \{0, 1\}^n$ be the vector with a unique nonzero entry $(\mat{e}_i)_i = 1$.

\begin{propn}\label{propn: basic}
	For a multilinear polynomial $f(Z) = f(Z_1, \ldots, Z_n)$, we have
	\[f(Z) - f(Z^{(i)}) ~=~ (Z_i - \resamp{Z_i})\cdot \grad_{\mat{e}_i}f(Z)\]
\end{propn}

\begin{proof}[Proof of \cref{lem: main_rademacher}]
	Consider the Hermitian dilation $\herm{\mF}_{a, b} = \mF_{a, b} \oplus \mF_{a, b}^T$. Define $\herm{\mX}_{a, b} = \herm{\mF}_{a, b} - \EE \herm{\mF}_{a, b} = \mX_{a, b} \oplus \mX_{a, b}^T$. By \cref{thm: main_efron_stein} applied to $\herm{\mX}_{a, b}$,
	$\Etr{\herm{\mX}_{a, b}^{2t}} \le (2(2t - 1))^t \Etr{\mV_{a, b}^t}$
	where $\mV_{a, b}$ is the variance proxy
	$\mV_{a, b} = \frac{1}{2} \displaystyle\sum_{i = 1}^n\EE[(\herm{\mX}_{a, b} - \herm{\mX}^{(i)}_{a, b})^2 | Z]$.
    By a simple computation,
	$\Etr{\herm{\mX}_{a, b}^{2t}} = \Etr{(\mX_{a, b}\mX_{a, b}^\T)^t} + \Etr{(\mX_{a, b}^\T\mX_{a, b})^t} = 2 \Esch{\mX_{a, b}}{2t}$, therefore
	\begin{align*}
		\mV_{a, b} 
		&= \frac{1}{2}\sum_{i = 1}^n \EE\bigg[\begin{bmatrix}
			(\mX_{a, b} - \mX_{a, b}^{(i)})(\mX_{a, b} - \mX_{a, b}^{(i)})^\T & 0\\
			0 & (\mX_{a, b} - \mX_{a, b}^{(i)})^\T(\mX_{a, b} - \mX_{a, b}^{(i)})
		\end{bmatrix}|Z\bigg]\\
		&= \frac{1}{2} \begin{bmatrix}
			\sum_{i = 1}^n\EE[(\mF_{a, b} - \mF_{a, b}^{(i)})(\mF_{a, b} - \mF_{a, b}^{(i)})^\T|Z] & 0\\
			0 & \sum_{i = 1}^n\EE[(\mF_{a, b} - \mF_{a, b}^{(i)})^\T(\mF_{a, b} - \mF_{a, b}^{(i)})|Z]
		\end{bmatrix}
	\end{align*}

	We will use the following claim that we will prove later.
	\begin{claim}\label{claim: reduction}
		We have the following relations.
		\[\sum_{i = 1}^n\EE[(\mF_{a, b} - \mF_{a, b}^{(i)})(\mF_{a, b} - \mF_{a, b}^{(i)})^\T|Z] = 2(b + 1)\mF_{a, b + 1}\mF_{a, b + 1}^\T\]
		\[\sum_{i = 1}^n\EE[(\mF_{a, b} - \mF_{a, b}^{(i)})^\T(\mF_{a, b} - \mF_{a, b}^{(i)})|Z] = 2(a + 1)\mF_{a + 1, b}^\T\mF_{a + 1, b}\]
	\end{claim}
	This gives $\Etr{\mV_{a, b}^t} = (b + 1)^t\Esch{\mF_{a, b + 1}}{2t} + (a + 1)^t\Esch{\mF_{a + 1, b}}{2t}$. Therefore, we get
	\begin{align*}
		2 \Esch{\mX_{a, b}}{2t} &= \Etr{\herm{\mX}_{a, b}^{2t}}\\
		&\le (2(2t - 1))^t \Etr{\mV_{a, b}^t}\\
		&= (2(2t - 1))^t((b + 1)^t\Esch{\mF_{a, b + 1}}{2t} + (a + 1)^t\Esch{\mF_{a + 1, b}}{2t})\\
		&= (2(2t - 1))^t((b + 1)^t\Esch{\mX_{a, b + 1} + \EE\mF_{a, b + 1}}{2t} + (a + 1)^t\Esch{\mX_{a + 1, b} + \EE\mF_{a + 1, b}}{2t})\\
		&\le (16t)^t((b + 1)^t(\Esch{\mX_{a, b + 1}}{2t} + \sch{\EE\mF_{a, b + 1}}{2t}) + (a + 1)^t(\Esch{\mX_{a + 1, b}}{2t} + \sch{\EE\mF_{a + 1, b}}{2t})\\
		&\le (16t\dpoly)^t(\Esch{\mX_{a, b + 1}}{2t} + \sch{\EE\mF_{a, b + 1}}{2t} + \Esch{\mX_{a + 1, b}}{2t} + \sch{\EE\mF_{a + 1, b}}{2t})
	\end{align*}
\end{proof}

It remains to prove the claim.
\begin{proof}[Proof of~\cref{claim: reduction}]
	We will prove the first equality. The second one is analogous.
	For $I \in \cI, J \in \cJ, \al, \beta \in \{0, 1\}^n$, we have
	\[(\mF_{a, b} - \mF^{(i)}_{a, b})[(I, \al), (J, \beta)] = \begin{dcases}
		\grad_{\al + \beta} (f_{I, J}(Z) - f_{I, J}(Z^{(i)})) & \text{ if $|\al|_0 = a, |\beta|_0 = b, \al\cdot \beta = 0$}\\
		0 & \text{o.w.}
	\end{dcases}\]
	By \cref{propn: basic}, the first expression simplifies to $(Z_i - \resamp{Z_i})\grad_{\mat{e}_i}\grad_{\al + \beta} f_{I, J}(Z)$. Define the matrix $\mF_{a, b, i}$ to be the matrix with the same set of rows and columns as $\mF_{a, b}$ and whose only nonzero entries are given by
	\[\mF_{a, b, i}[(I, \al), (J, \beta + \mat{e}_i)] = \grad_{\mat{e}_i}\grad_{\al + \beta} f_{I, J}(Z) \text{ if $|\al|_0 = a, |\beta|_0 = b, \beta \cdot \mat{e}_i = 0, \al\cdot (\beta + \mat{e}_i) = 0$}\]

	Then, it's easy to see that $\sum_{i = 1}^n \mF_{a, b, i}\mF_{a, b, i}^\T = (b + 1)\mF_{a, b + 1}\mF_{a, b + 1}^T$ and $(\mF_{a, b} - \mF_{a, b}^{(i)})(\mF_{a, b} - \mF_{a, b}^{(i)})^\T = (Z - \resamp{Z_i})^2\mF_{a, b, i}\mF_{a, b, i}^\T$. The latter equality implies
	\[\EE[(\mF_{a, b} - \mF_{a, b}^{(i)})(\mF_{a, b} - \mF_{a, b}^{(i)})^\T|Z] = \EE[(Z_i - \resamp{Z_i})^2\mF_{a, b, i}\mF_{a, b, i}^\T|Z] = 2\mF_{a, b, i}\mF_{a, b, i}^\T\]

	Therefore,
	\[\sum_{i = 1}^n\EE[(\mF_{a, b} - \mF_{a, b}^{(i)})(\mF_{a, b} - \mF_{a, b}^{(i)})^\T|Z] = 2\sum_{i = 1}^n\mF_{a, b, i}\mF_{a, b, i}^\T = 2(b + 1)\mF_{a, b + 1}\mF_{a, b + 1}^\T\]
\end{proof}



%% file: tensor_network_norm_bound.tex

We consider the following result from \cite{hopkins2015tensor, hopkins2018statistical}. We remark that this result could also be obtained via other standard techniques, but we showcase it as it serves as a simple warm-up to familiarize the reader with our method.

\begin{lemma}[\cite{hopkins2018statistical}, Theorem 6.7.1]
	Let $c \in \{1, 2\}$ and let $d \ge 1$ be an integer. Let $\mA_1, \ldots, \mA_{n^c}$ be i.i.d. random matrices uniformly sampled from $\pmone^{n^d \times n^d}$. Then, with probability $1 - O(n^{-100})$,
	\[\norm{\sum_{k \le n^c} \mA_k \otimes \mA_k- \EE \sum_{k \le n^c} \mA_k \otimes \mA_k} \le C\sqrt{d}n^{(2d + c) / 2} (\log n)^{1/2}\]
	for an absolute constant $C > 0$.
\end{lemma}

Using our framework, we will prove a slightly relaxed version of the inequality where $\sqrt{d} (\log n)^{1/2}$ is replaced by $\log n$, while not losing on the dominating term $n^{(2d + c)/2}$.
We remark that we have not attempted to optimize these extra factors in front of the dominating term $n^{(2d + c)/2}$, so it's plausible that a more careful analysis can obtain a slightly better bound.

\begin{proof}[Proof of the relaxed bound]
	Let the $i, j$-th entry of $\mA_k$ be $a_{k, i, j}$.
	Let $\mF = \sum_{i \le n^c} \mA_k \otimes \mA_k - \EE \sum_{i \le n^c} \mA_k \otimes \mA_k$ be a random matrix on the variables $a_{k, i, j}$ for $k \le n^c, i, j \le n^d$. So $\EE \mF = 0$ and we are looking for bounds on $\norm{\mF}$. The entries are given by
	\[\mF[(i_1, i_2), (j_1, j_2)] = \begin{dcases}
		\sum_{k \le n^c} a_{k, i_1, j_1}a_{k, i_2, j_2} & \text{ if $(i_1, j_1) \neq (i_2, j_2)$}\\
		0 & \text{ if $(i_1, j_1) = (i_2, j_2)$}
	\end{dcases}\]
	The nonzero entries are homogeneous polynomials of degree $2$. Using \cref{thm: main_rademacher},
	\[\Esch{\mF}{2t} \le (32t)^{2t}(\sch{\EE\mF_{2, 0}}{2t} + \sch{\EE\mF_{1, 1}}{2t} + \sch{\EE\mF_{0, 2}}{2t})\]




	We will consider each of these terms.
	In the following arguments, we restrict attention to indices $i_1, i_2, j_1, j_2$ such that $(i_1, j_1) \neq (i_2, j_2)$.

	\begin{enumerate}
		\item $\EE\mF_{2, 0}$ has nonzero entries in row $((i_1, i_2), \{(k, i_1, j_1), (k, i_2, j_2)\})$ and column $(j_1, j_2)$ and all these entries are $1$.  The Schatten norm does not change when we permute the rows and columns. So, we can group the rows on $k, i_1, i_2$ and within each group, we can sort $j_1, j_2$ in both rows and columns. We get a matrix having $n^{2d + c}$ identity matrices, each of dimensions $n^{2d} \times n^{2d}$, stacked on top of each other. Using the definition, the Schatten-$2t$ norm of this matrix is easily computed to be $\sch{\EE\mF_{2, 0}}{2t}  = n^{c + 4d} n^{t(2d + c)}$.

	\item $\EE\mF_{1, 1}$ has nonzero entries in either row $((i_1, i_2), \{(k, i_1, j_1)\})$ and column $((j_1, j_2), \{(k, i_2, j_2)\})$; or row $((i_1, i_2), \{(k, i_2, j_2)\})$ and column $((j_1, j_2), \{(k, i_1, j_1)\})$ and all these entries are $1$. So we can write $\EE\mF_{1, 1} = \mA + \mB$ corresponding to the 2 sets of entries. Arguing just as in the previous case,  we can obtain $\sch{\mA}{2t} = n^{c + 4d} n^{t(2d + c)}$ where we group the rows on $k, i_2, j_1$ and $\sch{\mB}{2t} = n^{c + 4d} n^{t(2d + c)}$ where we group the rows on $k, i_1, j_2$.
    Therefore,	$\sch{\EE\mF_{1, 1}}{2t} \le 2^{2t} (\sch{\mA}{2t} + \sch{\mB}{2t}) = 2^{2t + 1} n^{c + 4d} n^{t(2d + c)}$.


	\item The case $\EE\mF_{0, 2}$ is identical to $\EE\mF_{2, 0}$.
	\end{enumerate}

	Putting them together, $\Esch{\mF}{2t} \le (C't)^{2t} n^{c + 4d}n^{t(2d + c)}$
	for an absolute constant $C' > 0$. Now, we apply Markov's inequality to get
	\begin{align*}
		Pr[\norm{\mF - \EE \mF} \ge \theta] ~\le~ Pr[\sch{\mF - \EE \mF}{2t} \ge \theta^{2t}] ~\le~ \theta^{-2t} \EE\sch{\mF - \EE \mF}{2t} ~\le~ \theta^{-2t}(C't)^{2t} n^{c + 4d}n^{t(2d + c)}
	\end{align*}
	We now set $\theta = \eps^{-1/(2t)} (C't)n^{(c+4d)/t} n^{(2d + c)/2}$ to make this expression at most $\eps$. Plug in $\eps = n^{-100}$ and set $t = \log n$ to obtain that $\norm{\mF - \EE \mF} \le Cn^{(2d + c) / 2} \log n$ holds with probability $1 - n^{-100}$, where $C > 0$ is an absolute constant.
\end{proof}



%% file: dense_graph_matrices.tex
In this section, we first define graph matrices and then show how to obtain norm bounds for \textit{dense graph matrices}, i.e. the case when $G \sim \cG_{n, 1/2}$, using our framework. Handling \textit{sparse graph matrices}, i.e. the case when $G \sim \cG_{n, p}$ for $p = o(1)$, may not work well with our basic framework as we will explain in \cref{sec: failure_of_basic}. Instead, our general framework in \cref{sec: general_recursion} will handle this case well and we obtain sparse graph matrix norm bounds in \cref{sec: sparse_graph_matrices}.


\subsubsection{Definitions}

Define by $\cG_{n, p}$ the \Erdos-\Renyi random graph on the vertex set $[n]$ with $n$ vertices, where each edge is present independently with probability $p$. Let the graph be encoded by variables $G_{i, j} \in \Omega = \{-\sqrt{\frac{1 - p}{p}}, \sqrt{\frac{p}{1 - p}}\}$ where $-\sqrt{\frac{1 - p}{p}}$ indicates the presence of the edge $\{i, j\}$ and $\sqrt{\frac{p}{1 - p}}$ indicates absence, for all $1 \le i, j \le n$.

So, each $G_{i, j}$ for $i < j$ is sampled from $\Omega$ where $G_{i, j}$ takes the value $-\sqrt{\frac{1 - p}{p}}$ with probability $p$ and takes the value $\sqrt{\frac{p}{1 - p}}$ otherwise. Here, $\Omega$ has been normalized so that $\EE_{x \sim \Omega}[x] = 0, \EE_{x \sim \Omega}[x^2] = 1$. as is standard in $p$-biased Fourier analysis.

When $p = \nicefrac{1}{2}$, we are in the setting of \textit{dense graph matrices}. Then, $\cG_{n, 1/2}$ can be thought of as a sampling of the $G_{i, j}, i < j$ independently and uniformly from $\Omega = \{-1, 1\}$.
For a set of edges $E \subseteq \binom{[n]}{2}$, define $G_E := \prod_{e \in E} G_e$. When $p = \nicefrac{1}{2}$, the $G_E$ correspond to the Fourier basis for functions of the graph.

Define $\cI$ to be the set of sub-tuples of $[n]$, including the empty tuple. Graph matrices will have rows and columns indexed by $\cI$. Each graph matrix has a succinct representation as a graph with some extra information, that is called a \textit{shape}.

\begin{definition}[Shape]
	A shape is a tuple $\tau = (V(\tau), E(\tau), U_{\tau}, V_{\tau})$ where $(V(\tau), E(\tau))$ is a graph and $U_{\tau}, V_{\tau}$ are ordered subsets of the vertices.
\end{definition}

\begin{definition}[Realization]
	Given a shape $\tau$, a realization of $\tau$ is an injective map $\varphi: V(\tau) \to [n].$
\end{definition}

\begin{definition}[Graph matrices]
	Let $\tau$ be a shape.
	The graph matrix $\graphmat{\tau} \, : \, \{ \pm 1\}^{n \choose 2} \rightarrow \R^{\cI\times \cI}$ is defined to be the matrix-valued function with $I, J$-th entry defined as follows.
	\[
	\mM_{\tau}[I, J] := \sum_{\substack{\text{Realization }\phi\\ \phi(U_{\tau}) = I, \phi(V_{\tau}) = J}}{G_{\phi(E(\tau))}} = \sum_{\substack{\text{Realization }\phi\\ \phi(U_{\tau}) = I, \phi(V_{\tau}) = J}}\prod_{(u, v) \in E(\tau)} G_{\phi(u), \phi(v)}
	\]
	In other words, we sum over all realizations of $\tau$ that map $U_{\tau}, V_{\tau}$ to $I, J$ respectively and for each such realization, we have a term corresponding to the Fourier character that the realization gives.
\end{definition}

\begin{figure}[!h]
	\centering
	\includegraphics[trim={2cm 21cm 0 1cm}, clip, scale=1]{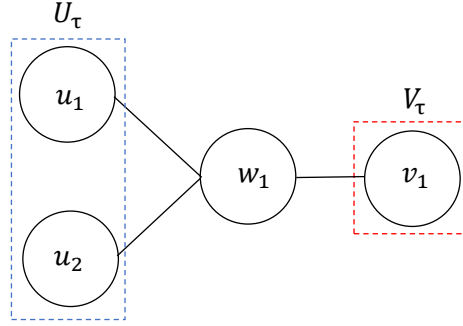}
	\caption{Left: Shape corresponding to adjacency matrix, Right: Example of a more complicated shape}
	\label{fig: shape}
\end{figure}

The following examples illustrate some simple graph matrices.

\begin{example}[Adjacency matrix]
	Let $\tau$ be the shape on the left in \cref{fig: shape}, with two vertices $V(\tau) = \{u,v\}$ and a single edge $E(\tau) = \{\{ u,v\}\}$. $U_\tau, V_\tau$ are $(u), (v)$ respectively where we use tuples to indicate ordering.
	Then $\mM_\tau$ has nonzero entries $\mM_\tau[(i), (j)](G) = G_{i, j}$ for all $i \neq j$.
	If $G \in \{ \pm 1\}^{n \choose 2}$ is thought of as a graph, then $\mM_\tau$ has as principal submatrix the $\pm 1$ adjacency matrix of $G$ with zeros on the diagonal, and the other entries are $0$.
\end{example}

\begin{example}
	In \cref{fig: shape}, consider the shape $\tau$ on the right. We have $U_{\tau} = (u_1, u_2), V_{\tau} = (v_1), V(\tau) = \{u_1, u_2, v_1, w_1\}$ and $E(\tau) = \{\{u_1, w_1\}, \{u_2, w_1\}, \{w_1, v_1\}\}$. $\mM_{\tau}$ is a matrix with rows and columns indexed by sub-tuples of $[n]$. Its nonzero entries are in rows $I$ and columns $J$ with $|I| = |U_{\tau}| = 2$ and $|J| = |V_{\tau}| = 1$ respectively. Specifically, for all distinct $a_1, a_2, b_1$, the entry corresponding to row $(a_1, a_2)$ and column $(b_1)$ is $\sum_{c_1 \in [n] \setminus \{a_1, a_2, b_1\}} G_{a_1, c_1}G_{a_2, c_1}G_{c_1, b_1}$.
	Here, each term is obtained via the realization $\phi$ that maps $u_1, u_2, w_1, v_1$ to $a_1, a_2, c_1, b_1$ respectively. Succinctly, \[\mM_{\tau} =
	\begin{blockarray}{rl@{}c@{}r}
		& & \makebox[0pt]{column $(b_1)$} \\[-0.5ex]
		& & \,\downarrow \\[-0.5ex]
		\begin{block}{r(l@{}c@{}r)}
			&  & \vdots & \\[-0.2ex]
			\text{row }(a_1, a_2) \to \mkern-9mu & \raisebox{0.5ex}{\makebox[3.2em][l]{\dotfill}} & \sum_{c_1 \in [n] \setminus \{a_1, a_2, b_1\}} G_{a_1, c_1}G_{a_2, c_1}G_{c_1, b_1} & \raisebox{0.5ex}{\makebox[4.2em][r]{\dotfill}} \\[+.5ex]
			&  & \vdots & \\
		\end{block}
	\end{blockarray}\]
\end{example}


Intuitively, graph matrices are symmetrizations of the Fourier basis, where the symmetry is incorporated by summing over all realizations of ``free'' vertices $V(\tau) \setminus U_{\tau} \setminus V_{\tau}$ of the shape $\tau$.
For more examples of graph matrices and why they can be a useful tool to work with, see \cite{ahn2016graph}.

\subsubsection{Norm bounds for dense graph matrices}\label{sec: norm_bounds_for_dense_graph_matrices}

In this section, we study the concentration of the so-called ``dense graph matrices'' which is a term that refers to graph matrices $M_{\tau}$ in the setting $p = \nicefrac{1}{2}$.
Since the edges of a random graph sampled from $\cG_{n,1/2}$ can be viewed as independent Rademacher random variables, we can apply our framework in this setting.
%


In particular, we will obtain bounds on $\Esch{\mM_{\tau} - \EE\mM_{\tau}}{2t}$.
The $G_{i, j} \in \{-1, 1\}$ correspond to the $Z_i$s in \cref{sec: basic_recursion} and for a fixed shape $\tau$, $\mM_{\tau}$ will be the matrix $\mF$ we are interested in analyzing. For $I, J \in \cI$, $\mM_{\tau}[I, J]$ is a nonzero polynomial only when there exists at least one realization of $\tau$ that maps $U_{\tau}, V_{\tau}$ to $I, J$ respectively. In particular, we must have $|I| = |U_{\tau}|$ and $|J| = |V_{\tau}|$. In this case, $\mM_{\tau}[I, J]$ is a homogenous polynomial of degree $|E(\tau)|$.
By \cref{thm: main_rademacher}, we have
\[\Esch{\mM_{\tau} - \EE\mM_{\tau}}{2t} ~\le~ \sum_{a + b \ge 1\atop a, b \ge 0}(16t|E(\tau)|)^{(a + b)t}\sch{\EE\mM_{\tau, a, b}}{2t}\]
where for integers $a, b \ge 0$, $\mM_{\tau, a, b}$ is defined to be the matrix with rows and columns each indexed by $\cI \times \{0, 1\}^{\binom{n}{2}}$ such that for all $I, J \in \cI$, we have
\[\mM_{\tau, a, b}[(I, \al), (J, \beta)] ~=~ \begin{dcases}
	\grad_{\al + \beta} \mM_{\tau}[I, J] & \text{ if $|\al|_0 = a, |\beta|_0 = b, \al \cdot \beta = 0$}\\
	0 & \text{o.w.}
\end{dcases}
\]

For any multilinear homogenous polynomial $f$ of degree $d$, since $\EE[G_{i, j}] = 0$ for all $i, j$, we have $\grad_{\al}f = 0$ whenever $|\al|_0 < d$. Therefore, $\EE\mM_{\tau, a, b} = 0$ for all $a + b < |E(\tau)|$. Moreover, $\EE\mM_{\tau, a, b} = 0$ whenever $a + b \neq |E(G)|$ otherwise $\EE\mM_{\tau, a, b} = \mM_{\tau, a, b}$. So, we can further simplify the above expression to
\[\Esch{\mM_{\tau} - \EE\mM_{\tau}}{2t} ~\le~ \sum_{a + b = |E(\tau)|\atop a, b \ge 0}(16t|E(\tau)|)^{|E(\tau)|t}\sch{\mM_{\tau, a, b}}{2t}\]

It remains to analyze $\sch{\mM_{\tau, a, b}}{2t}$ for $a + b = |E(\tau)|$. We will see that analyzing these matrices is much simpler since they are deterministic matrices and simple computations using the Frobenius norm bound will work well. To state our final bounds, we need to define the notion of vertex separators of shapes.

\begin{remark}
	As we will see, when analyzing the Frobenius norms for these deterministic matrices, the notion of the minimum vertex separator arises naturally. In prior trace method calculations (e.g. \cite{medarametla2016bounds}, \cite{ahn2016graph}), this required ingenious combinatorial observations.
\end{remark}

\begin{restatable}[Vertex separator]{definition}{vertexseparator}
	For a shape $\tau$, define a vertex separator to be a subset of vertices $S \subseteq V(\tau)$ such that there is no path from $U_{\tau}$ to $V_{\tau}$ in $\tau \setminus S$, which is the shape obtained by deleting all the vertices of $S$ (including all edges they're incident on).
\end{restatable}

For a shape $\tau$, denote by $S_{\tau}$ a vertex separator of the smallest size. Also, let $I_{\tau}$ be the set of isolated vertices (vertices with degree $0$) in $V(\tau) \setminus U_{\tau} \setminus V_{\tau}$, so the presence of these vertices essentially scale the matrix by a scalar factor.

\begin{theorem}\label{thm: dense_graph_matrix_norm_bounds}
	For a shape $\tau$ and any integer $t \ge 1$,
	\[\EE\sch{\mM_{\tau} - \EE \mM_{\tau}}{2t} \le \bigg(C^{t|E(\tau)|}n^{|V(\tau)|} t^{t|E(\tau)|}|E(\tau)|^{2t|E(\tau)|}\bigg)n^{t(|V(\tau)| - |S_{\tau}| + |I_{\tau}|)}\]
	for an absolute constant $C > 0$.
\end{theorem}

Up to lower order terms, the same result has been shown before in \cite{medarametla2016bounds, ahn2016graph}. To interpret this bound, assume that $\tau$ has a constant number of vertices. By setting $t \approx \polylog(n)$, we get \[\norm{\mM_{\tau}} = \widetilde{\bigoh}\left(\sqrt{n}^{|V(\tau)| - |S_{\tau}| + |I_{\tau}|}\right)\] with high probability, where $\widetilde{\bigoh}$ hides logarithmic factors.
This is obtained by applying Markov's inequality on the bound on $\Esch{\mM_{\tau}}{2t}$. If $\tau$ has at least one edge, then $\EE \mM_{\tau} = 0$ and \cref{thm: dense_graph_matrix_norm_bounds} yields such bounds. If $\tau$ has no edges, then it's quite simple to obtain such a bound and we include it in \cref{lem: empty_shape} for the sake of completeness.

\cref{cor: dense_graph_matrix_norm_bounds} makes precise the high probability bound above. Therefore, this power of $n$ is essentially what controls the norm bound and this is utilized heavily in applications (e.g. \cite{barak2019nearly, ghosh2020sum, potechin2020machinery}).

\begin{proof}[Proof of \cref{thm: dense_graph_matrix_norm_bounds}]
    We first argue that we can assume $I_{\tau} = \emptyset$. This is because of the following reason. Each distinct vertex in $\tau$ of degree $0$ essentially scales the matrix by a factor of at most $n$. And in the right hand side of the inequality, each vertex in $I_{\tau}$ contributes a factor of $n^{2t}$ accordingly, from $n^{t|V(\tau)|}$ and from $n^{t|I_{\tau}|}$, and the other changes only weaken the inequality.

	Now, fix $a, b \ge 0$ such that $a + b = |E(\tau)|$ and consider $\mM_{\tau, a, b}$. For $I, J \in \cI, \al, \beta \in \{0, 1\}^{\binom{n}{2}}$ such that $|\al|_0 = a, |\beta|_0 = b, \al \cdot \beta = 0$, by definition,
    \begin{align*}
        \mM_{\tau, a, b}[(I, \al), (J, \beta)] &~=~ \grad_{\al + \beta} \left(\sum_{\phi: \phi(U_{\tau}) = I, \phi(V_{\tau})= J} \prod_{u, v \in E(\tau)} G_{\phi(u), \phi(v)}\right)\\
        &~=~ |\{\phi ~|~ \phi(U_{\tau}) = I, \phi(V_{\tau})= J, \phi(E(\tau)) = \supp(\al + \beta)\}|
    \end{align*}
    where $\supp(.)$ denotes the support. We will now obtain norm bounds on these deterministic matrices by reinterpreting them as graph matrices for different shapes.
    Let $P = (E_1, E_2)$ denote the partition of $E(\tau) = E_1 \sqcup E_2$ into two ordered sets $E_1, E_2$, where $\sqcup$ denotes disjoint union. Let the set of ordered partitions $P$ be $\cP$. Then, we can write $\mM_{\tau, a, b} = \sum_{P \in \cP} \mM_{\tau, a, b, P}$ where
    \[\mM_{\tau, a, b, P}[(I, \al), (J, \beta)] ~=~ |\{\phi ~|~ \phi(U_{\tau}) = I, \phi(V_{\tau})= J, \phi(E_1) = \supp(\al), \phi(E_2) = \supp(\beta)\}|\]

    Also, $|\cP| \le (4|E(\tau)|)^{|E(\tau)|}$ and so, by \cref{fact: holder},
    \[\sch{\mM_{\tau, a, b}}{2t} \le (4|E(\tau)|)^{t|E(\tau)|} \sum_{P \in \cP}\sch{\mM_{\tau, a, b, P}}{2t}\]


    Each $\mM_{\tau, a, b, P}$ can be interpreted as a graph matrix for a different shape $\tau_P$, with the same vertex set and no edges. Let $V(\tau_P) = V(\tau), E(\tau_P) = \emptyset$ and set $U_{\tau_P} = U_{\tau} \cup V(E_1), V(\tau_P) = V_{\tau} \cup V(E_2)$ using a canonical ordering. Then, $\mM_{\tau, a, b}$ is equal to $\mM_{\tau_P}$ up to renaming of the rows and columns. For an illustration, see \cref{fig: evolution_new}.


	\begin{figure}[!h]
		\centering
		\includegraphics[trim={2cm 20cm 0 2cm}, clip, scale=0.9]{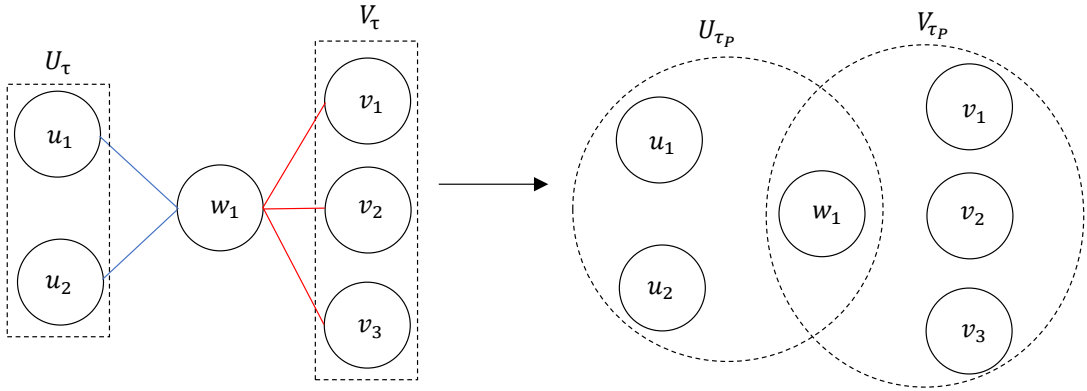}
		\caption{An example illustrating how $\tau_P$ is defined. In this example, $P$ constraints the blue and red edges to go to $\al$ and $\beta$ respectively. $U_{\tau_P}, V_{\tau_P}$ have an ordering on the vertices (not shown here).}
		\label{fig: evolution_new}
	\end{figure}

	This graph matrix has a block diagonal structure indexed by the realizations of the set of common vertices $S = U_{\tau_P} \cap V_{\tau_P}$. Indeed, for $K \in [n]^S$, let $\mM_{\tau_P, K}$ be the block of $\mM_{\tau_P}$ with $\phi(S) = K$. Then, $\mM_{\tau_P, K}\mM_{\tau_P, K'}^\T = \mM_{\tau_P, K}^\T\mM_{\tau_P, K'} = 0$ for $K \neq K'$ and so,
	\begin{align*}
		\Esch{\mM_{\tau, a, b}}{2t} \le (4|E(\tau)|)^{t|E(\tau)|} \sum_{P \in \cP}\sch{\mM_{\tau_P}}{2t}
		&= (4|E(\tau)|)^{t|E(\tau)|} \sum_{P \in \cP} \sum_{T \in [n]^S}\sch{\mM_{\tau_P, T}}{2t}\\
		&\le (4|E(\tau)|)^{t|E(\tau)|} \sum_{P \in \cP} \sum_{T \in [n]^S}\left(\sch{\mM_{\tau_P, T}}{2}\right)^t
	\end{align*}
	where we bounded the Schatten norm by the appropriate power of the Frobenius norm.
	For any fixed $K \in [n]^S$, the entries of $\mM_{\tau_P, K}$ take values in $\{0, 1\}$ and the number of nonzero entries is at most $n^{|V(\tau)| - |S|}$ because the realizations of vertices in $S$ are fixed and the other vertices have at most $n$ choices each. Therefore, $\sch{\mM_{\tau_P, K}}{2} \le n^{|V(\tau)| - |S|}$.

	Finally, we bound $|S|$ to estimate how large this term can be over all possibilities of $P$.
    We argue that $S$ blocks all paths from $U_{\tau}$ to $V_{\tau}$. To see this, consider any path from $U_{\tau}$ to $V_{\tau}$, it must contain an edge $(u, v) \in E(\tau)$ such that $u \in U_{\tau_P}, v \in V_{\tau_P}$. We must either have $(u, v) \in E_1$, in which case $u,  v \in U_{\tau_P}$ and $v \in S$, or $(u, v) \in E_2$, in which case $u, v \in V_{\tau_P}$ and $u \in S$. In either case, $S$ must contain either $u$ or $v$. This argument implies $S$ must be a vertex separator of $\tau$, giving $|S| \ge |S_{\tau}|$.
    For a proof by picture, see \cref{fig: proof_by_picture}.


    \begin{figure}[!h]
        \centering
        \includegraphics[trim={2cm 20cm 2cm 2cm}, clip, scale=0.9]{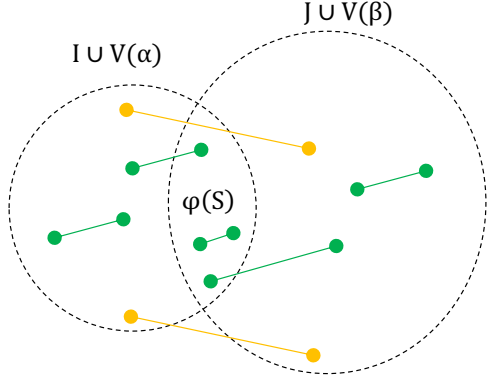}
        \caption{Proof by picture that $|S| \ge |S_{\tau}|$. Green edges can occur in $\tau$, orange edges cannot, so $S$ blocks all paths from $U_{\tau}$ to $V_{\tau}$.}
        \label{fig: proof_by_picture}
    \end{figure}

	We also have the trivial upper bound $|S| \le |V(\tau)|$. Ultimately, this gives
	\begin{align*}
		\sch{\mM_{\tau, a, b}}{2t} &\le (4|E(\tau)|)^{t|E(\tau)|} \sum_{P \in \cP} \sum_{T \in [n]^S}n^{t(|V(\tau)| - |S_{\tau}|)}
		&\le (4|E(\tau)|)^{t|E(\tau)|}(4|E(\tau)|)^{|E(\tau)|}n^{|V(\tau)|}n^{t(|V(\tau)| - |S_{\tau}|)}
	\end{align*}
	Along with our prior discussion, we get
	\begin{align*}
		\Esch{\mM_{\tau} - \EE\mM_{\tau}}{2t} &\le 	\sum_{a + b = |E(\tau)|}(16t|E(\tau)|)^{|E(\tau)|t}\sch{\mM_{\tau, a, b}}{2t}\\
		&\le \sum_{a + b = 	|E(\tau)|}(16t|E(\tau)|)^{|E(\tau)|t}(4|E(\tau)|)^{t|E(\tau)|}(4|E(\tau)|)^{|E(\tau)|}n^{|V(\tau)|}n^{t(|V(\tau)| - |S_{\tau}|)}\\
		&\le \bigg(C^{t|E(\tau)|}n^{|V(\tau)|} t^{t|E(\tau)|}|E(\tau)|^{2t|E(\tau)|}\bigg)n^{t(|V(\tau)| - |S_{\tau}|)}
	\end{align*}
	for an absolute constant $C > 0$.
\end{proof}

\begin{remark}
Note that while the proof of the norm bound above still requires some combinatorial analysis, this
arises mostly from a mechanical application of the general result \cref{thm: main_rademacher}. Also,
one only needs the simpler combinatorics of the fixed-size shapes obtained from the given shape $\tau$, rather
than increasingly large shapes formed by combining copies of $\tau$, as in the application of trace
method~\cite{ahn2016graph}.
\end{remark}

In the proof above, our analysis of the shape $\tau_P$ which has no edges, applies in general to any shape $\tau$ with no edges. For the sake of completeness, we state it explicity in the following lemma.

\begin{lemma}\label{lem: empty_shape}
	For a shape $\tau$ with no edges and any integer $t \ge 1$,
	$\Esch{\mM_{\tau}}{2t} \le n^{|U_{\tau} \cap V_{\tau}|}n^{t(V(\tau) - |U_{\tau} \cap V_{\tau}| + |I_{\tau}|)}$.
\end{lemma}

Note that this has the same form as \cref{thm: dense_graph_matrix_norm_bounds} because for a shape $\tau$ with no edges, the minimum vertex separator $S_{\tau}$ is just $U_{\tau} \cap V_{\tau}$.
The following corollary obtains high probability norm bounds for norms of graph matrices via Markov's inequality.

\begin{corollary}\label{cor: dense_graph_matrix_norm_bounds}
	For a shape $\tau$, for any constant $\eps > 0$, with probability $1 - \eps$,
	\[\norm{\mM_{\tau}} \le (C|E(\tau)| \log(n^{|V(\tau)|}/\eps))^{|E(\tau)|}\cdot\sqrt{n}^{|V(\tau)| - |S_{\tau}| + |I_{\tau}|}\]
	for an absolute constant $C > 0$.
\end{corollary}

\begin{proof}
	If $E(\tau) = \emptyset$, we invoke \cref{lem: empty_shape}. Otherwise, $\EE\mM_{\tau} = 0$ and we invoke \cref{thm: dense_graph_matrix_norm_bounds}. By an application of Markov's inequality,
	\begin{align*}
		Pr[\norm{\mM_{\tau}} \ge \theta] \le Pr[\sch{\mM_{\tau}}{2t} \ge \theta^{2t}]
		&\le \theta^{-2t} \EE\sch{\mM_{\tau}}{2t}\\
		&\le \theta^{-2t}\bigg((C')^{t|E(\tau)|}n^{|V(\tau)|} t^{t|E(\tau)|}|E(\tau)|^{2t|E(\tau)|}\bigg)n^{t(|V(\tau)| - |S_{\tau}| + |I_{\tau}|)}
	\end{align*}
	for an absolute constant $C' > 0$. To make this expression at most $\eps$, we simply set
	$\theta = \bigg(\eps^{-1/(2t)} (C'')^{|E(\tau)|} n^{|V(\tau)|/(2t)}
    t^{|E(\tau)|/2}|E(\tau)|^{|E(\tau)|}\bigg)
    \sqrt{n}^{|V(\tau)| - |S_{\tau}| + |I_{\tau}|}$ for an absolute constant $C'' > 0$. Finally, set $t = \frac{1}{2} \log(n^{|V(\tau)|}/\eps)$ to complete the proof.
\end{proof}



%% file: failure_of_basic.tex
In this section, we elaborate on the difficulties that arise when working with random variables that are not necessarily Rademacher. In this case, note that we cannot assume that the polynomial entries are multilinear as well.

To recall the setting, we are given a random matrix $\mF$ whose entries are low degree polynomials in random variables $Z_1, \ldots, Z_n$ which are independently sampled from arbitrary distributions. And we wish to obtain concentration bounds on how much $\mF$ can deviate from its mean, by way of controlling $\Esch{\mF- \EE\mF}{2t}$.

Building on the ideas from \cref{sec: basic_recursion}, we could attempt to use matrix Efron-Stein, \cref{thm: main_efron_stein} and hope to obtain a similar recursion framework. We now discuss what happens if we do this. Assume $\EE[Z_i] = 0, \EE[Z_i^2] = 1$. We can proceed similar to the proof of \cref{thm: main_rademacher}. So, we consider $\mX$ as a principal submatrix of $\mX_{0, 0}$ and follow through \cref{lem: main_rademacher}. The main change will happen in \cref{claim: reduction}. In particular, the equation $\EE[(Z_i - \resamp{Z_i})^2|Z] = 2$ is no longer true. Instead, we will have $\EE[(Z_i - \resamp{Z_i})^2|Z] = 1 + Z_i^2$. So, we get the expression
\[\sum_{i = 1}^n (1 + Z_i^2)\mF_{a, b, i}\mF_{a, b, i}^\T = \sum_{i = 1}^n \mF_{a, b, i}\mF_{a, b, i}^\T + \sum_{i = 1}^n Z_i^2\mF_{a, b, i}\mF_{a, b, i}^\T\]

The first term can been handled just as in the basic framework. Unfortunately, the second term will be a source of difficulty. To get around this difficulty, we could attempt to apply the matrix Efron-Stein inequality again on an appropriately constructed matrix.
To do this, we can interpret the second term as having been obtained after differentiating with respect to the variable $Z_i$ and then \textit{putting the variable back}. In contrast, we didn't need to put it back when working with Rademacher random variables.
But after we do this, when we recurse on these extra matrices, the new second term will contain the left hand side as a sub-term, thereby giving a trivial inequality and stalling the recursion.

To see this more clearly, consider the simplest case $a = b = 0$. Then, the first term $\sum_{i = 1}^n \mF_{a, b, i}\mF_{a, b, i}^\T$ will be equal to $\mF_{0, 1}\mF_{0, 1}^\T$ as we saw earlier. To evaluate the second term $\sum_{i = 1}^n Z_i^2\mF_{a, b, i}\mF_{a, b, i}^\T$ in a similar manner, we define the matrix $\mH$ to be the same as $\mF_{0, 1}$ except that each entry is now multiplied by $Z_i$ where $i$ is the differentiated variable in the column. That is, $\mH[I, (J, \mat{e}_i)] = Z_i\mF_{0, 1}[I, (J, \mat{e}_i)]$. Observe that in the definition of $\mH$, $Z_i$ has been put back after differentiating with respect to it. Then, the second term will be $\mH\mH^\T$ and we can hope to use Efron-Stein again on this matrix $\mH$ recursively.

We could do that and proceed similarly to the proof of \cref{lem: main_rademacher} with appropriate modifications as above. But since $\beta_i = 1$ already, differentiating with respect to $Z_i$ and putting it back, will return the same matrix $\mH$! So, we end up with an inequality of the form
\[\Esch{\mH}{2t} \le \bigoh(t)^{t}(\Esch{\mH}{2t} + \text{ other nonnegative terms})\]
Indeed, this is a tautology and will not be useful to us.

For a quick and dirty bound, suppose we had a parameter $L$ such that $1 + Z_i^2 \le L$ for our distributions, then we will be able to obtain a similar framework while incurring a loss of $\sqrt{L}$ at each step of the recursion. But unfortunately, this bound will be lossy. For example, if we do this computation for the centered normalized adjacency matrix of $G \sim \cG_{n, p}$, we will obtain a norm bound of $\widetilde{\bigoh}(\frac{\sqrt{n(1 - p)}}{\sqrt{p}})$ where $\widetilde{\bigoh}$ hides logarithmic factors.. This bound is tight for constant or even inverse polylogarithmic $p$. But for $p = n^{-\theta}$ for some constant $0 < \theta < 1$, this is not tight because in this regime, the true norm bound is known to be $\widetilde{\bigoh}(\sqrt{n})$ (see the early works of \cite{furedi1981eigenvalues, vu2005spectral} and for tighter bounds, see \cite{benaych2020spectral} and references therein).

If we dig into the details of what happened, this example illustrates that the matrix Efron-Stein inequality \cref{thm: main_efron_stein} becomes a tautology for certain kinds of matrices, that yield $\mV= \bigoh(1) \mX\mX^\T + \text{ other positive semidefinite matrices}$.

But in our framework in general, the aforementioned bad matrices occur when we differentiate with respect to variables that have already been differentiated on. In other words, the current definition of the variance proxy $\mV$ doesn't take into account whether we have already differentiated with respect to some variable $Z_i$. So, for the general recursion, we dive into the proof due to \cite{paulin2016} and modify it using structural properties of the intermediate matrices we obtain in our framework.

%% file: general_recursion.tex
We now assume $Z_1, \ldots, Z_n$ are i.i.d. random variables sampled from a distribution $\Omega$ with finite moments.
We assume that they are identically distributed for simplicity but our technique easily extends even when they are not identically distributed, as long as they are independent.
For each $i \le n$, define $\resamp{Z_i}$ to be an independent copy of $Z_i$ and define the vector $Z^{(i)} := (Z_1, \ldots, Z_{i - 1}, \resamp{Z_i}, Z_{i + 1}, \ldots, Z_n)$. Define $Z'$ to be the random vector defined by sampling $i$ from $[n]$ uniformly at random and then setting $Z' = Z^{(i)}$.

Let $\mF \in \RR[Z]^{\cI \times \cJ}$ be a matrix with rows and columns indexed by arbitrary sets $\cI, \cJ$ respectively such that for all $I \in \cI, J \in \cJ$, $\mF[I, J]$ are polynomials of $Z_1, \ldots, Z_n$. Let $\dpoly$ be the maximum degree of $\mF[I, J]$ over all entries $I, J$ and let $d$ be the maximum degree of $Z_i$ over all entries $\mF[I, J]$ and $i \le n$.


Similar to the Rademacher case, let $\mX := \mF- \EE\mF$. When the input is $Z$, we denote the matrices as $\mF, \mX$, etc and when the input is $Z^{(i)}$, denote the corresponding matrices as $\mF^{(i)}, \mX^{(i)}$, etc. In this section, we will give a general framework using which we can obtain bounds on $\Esch{\mF - \EE\mF}{2t}$ for any integer $t \ge 1$.
We set up a few preliminaries in order to state the main theorem.

\begin{definition}[Space $\cS$]
	Let $\cS$ be the space of mean-zero polynomials in $Z_1, \ldots, Z_n$ of degree at most $\dpoly$.
\end{definition}

For $\al \neq 0$, we also define the centered monomials
$\chi_{\al}(Z) = \prod_{\al_i > 0} (Z_i^{\al_i} - \EE[Z_i^{\al_i}])$.
By definition, $\chi_{\al} \in \cS$ for all $\al \neq 0, |\al|_1 \le \dpoly$. The following proposition is straightforward.

\begin{propn}\label{propn: basis}
	The set $\{\chi_{\al}(Z) | 1 \le |\al|_1 \le \dpoly\}$ forms a basis for $\cS$.
\end{propn}

For the general framework, we work over this basis because as we will see in \cref{sec: proof_of_general}, the ``inner kernel matrix'' is convenient to state in this basis.
The $\grad$ operator also works nicely with our polynomials $\chi_{\beta}$. Indeed, observe that $\grad_{\al}(\chi_{\beta}) = \begin{dcases}
	\chi_{\beta - \al} & \text{ if $\al \unlhd \beta$}\\
	0 & \text{ o.w.}
\end{dcases}$.

For a polynomial $f(Z)$ in $\cS$, denote by $\coef{f}{\al}$ the coefficient of $\chi_{\al}(Z)$ in the expansion of $f$, that is, $f(Z) = \displaystyle\sum_{0 \neq \al \in \NN^n} \coef{f}{\al}\chi_{\al}(Z)$.
We can naturally extend this notation to matrices that have mean $0$. So, we can write $\mX = \sum_{\al \neq 0} \coef{\mX}{\al} \chi_{\al}(Z)$ where $\coef{\mX}{\al}$ are deterministic matrices. In order to apply our recursion framework, we group this sum into terms based on $|\al|_0$. For $k \ge 1$, define $\mX_k = \sum_{|\al|_0 = k} \coef{\mX}{\al} \chi_{\al}(Z)$. Then,
$\mX = \sum_{k \ge 1} \mX_k$.
Note that when $k > \dpoly$, $\mX_k = 0$.

\begin{definition}[Indexing set $\cK$]
	We define $\cK \subseteq \NN^n \times \{0, 1\}^n$ to be the set of pairs $(\al, \gam)$ such that $|\al|_1 \le \dpoly, \al \in \NN^n$ and $\gam \le \al$ with $\gam \in \{0, 1\}^n$.
\end{definition}

\begin{remark}
    If we assume that the maximum degree of our polynomials $d_p$ is constant, then the size of $\cK$ is polynomially large, not exponentially large. Hence, the matrices we will consider below will also be of polynomial size when $d_p$ is constant.
\end{remark}

Define the diagonal matrices $\mD_1 \in \RR[Z]^{\cI \times \cK} \times \RR[Z]^{\cI \times \cK}$ and $\mD_2 \in \RR[Z]^{\cJ \times \cK} \times \RR[Z]^{\cJ \times \cK}$ with nonzero entries
\vspace{.1em}
\[\mD_1[(I, \al, \gam), (I, \al, \gam)] = \sqrt{\EE[Z^{2\al\cdot (1 - \gam)}]} Z^{\al\cdot \gam},\qquad \mD_2[(J, \al, \gam), (J, \al, \gam)] = \sqrt{\EE[Z^{2\al\cdot (1 - \gam)}]}Z^{\al\cdot \gam}\]

%

\begin{definition}[Matrices $\mG_{k, a, b}, \mF_{k, a, b}$]
	For integers $k, a, b$ such that $k \ge 1, a, b \ge 0$, define the matrix $\mG_{k, a, b}$ to have rows and columns indexed by $\cI \times\cK$ and $\cJ \times \cK$ respectively such that for all $(I, \al_1, \gam_1) \in \cI \times \cK$, $(J, \al_2, \gam_2) \in \cJ \times \cK$,
	\[\mG_{k, a, b}[(I, \al_1, \gam_1), (J, \al_2, \gam_2)] = \begin{dcases}
		\grad_{\al_1 + \al_2} \mX_k[I, J] & \text{ if $|\al_1|_0 = a, |\al_2|_0 = b, \al_1 \cdot \al_2 = 0$}\\
		0 & \text{o.w.}
	\end{dcases}
	\]
	Also, define $\mF_{k, a, b} := \mD_1 \mG_{k, a, b}\mD_2$.
\end{definition}

Note that when $k > \dpoly$, $\mF_{k, a, b} = 0$.

\begin{propn}
	For integers $k, a, b$ such that $k \ge 1, a, b \ge 0$, suppose $a + b < k$. Then each nonzero entry $f$ of $\mG_{k, a, b}$ has the property that $\coef{f}{\al}$ is nonzero only when $|\al|_0 = k - a - b$
\end{propn}

\begin{proof}
	The nonzero entries of $\mX_k$ only has terms containing exactly $k$ variables and $\grad_{\al_1 + \al_2}$ either zeroes out the term, or it truncates exactly $|\al_1 + \al_2|_0 = |\al_1|_0 + |\al_2|_0 = a + b$ variables.
\end{proof}

This also immediately implies that $\EE[\mG_{k, a, b}] = 0$ whenever $a + b < k$. Finally, when $k = a + b$, we have that $\mG_{k, a, b}$ is a deterministic matrix independent of the $Z_i$. These give rise to the matrices $\mF_{a + b, a, b}$ that appears in our main theorem.
We are now ready to state the main theorem.

\begin{theorem}[General recursion]
\label{thm: main_general}
Let the tuple of random variables $Z$ and the function $\mF$ be as above. Then, for all integers $t \ge 1$,
\[\Esch{\mF - \EE \mF}{2t} \le \sum_{a, b \ge 0, a + b \ge 1}(Ct^2d\dpoly^4)^{(a + b)t}\Esch{\mF_{a + b, a, b}}{2t}\]
	for an absolute constant $C > 0$.
\end{theorem}

Note that $\mF_{a + b, a, b} = \mD_1 \mG_{a + b, a, b} \mD_2$ where $\mD_1, \mD_2$ are diagonal matrices and $\mG_{a + b, a, b}$ is a deterministic matrix that's independent of $Z$. To analyze the expected Schatten norm of such matrices, we can resort to far simpler techniques. For instance, we can obtain a simple bound using an appropriate power of the Frobenius norm, and apply standard scalar concentration tools. We will see an example of this in \cref{sec: sparse_graph_matrices}.

\begin{remk}
	We have made no attempts to optimize the factors in front of the expectation in \cref{thm: main_general}, which we suspect can be improved. 
\end{remk}

We prove the main theorem by repeatedly applying the following technical lemma, the proof of which we defer to the next section.

\begin{restatable}{lemma}{maingeneral}\label{lem: main_general}
	For all integers $t \ge 1$, integers $k \ge 1, a, b \ge 0$ such that $a + b < k$,
	\[\Esch{\herm{\mF}_{k, a, b}}{2t} \le (Ct^2d\dpoly^2)^t (\Esch{\herm{\mF}_{k, a, b + 1}}{2t} + \Esch{\herm{\mF}_{k, a + 1, b}}{2t})\]
	for an absolute constant $C > 0$.
\end{restatable}

Using this lemma, we can complete the proof of the main theorem.

\begin{proof}[Proof of \cref{thm: main_general}]
	Using \cref{fact: holder}, we have $\Esch{\mX}{2t} \le \dpoly^{2t} \sum_{k = 1}^{\dpoly}\Esch{\mX_k}{2t}$. Note that for any $k \ge 1$, the matrix $\mX_k$ is a principal submatrix of $\mF_{k, 0, 0}$ with all other entries being $0$, so $\Esch{\mX_k}{2t} = \Esch{\mF_{k, 0, 0}}{2t} = \frac{1}{2} \Esch{\herm{\mF}_{k, 0, 0}}{2t}$. Therefore,
	$\Esch{\mX}{2t} \le \frac{1}{2}\dpoly^{2t} \displaystyle\sum_{k = 1}^{\dpoly}\Esch{\herm{\mF}_{k, 0, 0}}{2t}$.
	We now apply \cref{lem: main_general} repeatedly to all our terms until $k = a + b$, ultimately giving
	\[\Esch{\mX}{2t} \le \frac{1}{2}\dpoly^{2t}(Ct^2d\dpoly^2)^{(a + b)t} \sum_{a, b \ge 0, a + b \ge 1}\Esch{\herm{\mF}_{a + b, a, b}}{2t}\]
	Observing that $\Esch{\herm{\mF}_{a + b, a, b}}{2t} = 2\Esch{\mF_{a + b, a, b}}{2t}$ completes the proof.
\end{proof}



%% file: proof_of_general.tex
In this section, we will prove \cref{lem: main_general} using the high level strategy described in \cref{sec: intro}. This requires generalizing the results in \cite{paulin2016}, and the proof techniques may be of independent interest.

\subsection{Generalizing \cite{paulin2016} via explicit inner kernels}\label{sec: explicit_inner_kernels}

In our setting, observe that $(Z, Z')$ has the same distribution as $(Z', Z)$. This is what is known as an \textit{exchangeable pair} of variables, that will be extremely useful for our analysis. In particular, $Z, Z'$ have the same distribution and $\EE f(Z, Z') = \EE f(Z', Z)$ for every integrable function $f$.

\begin{definition}[Laplacian operator $\cL$]
	Define the operator $\cL$ on the space $\cS$ as
	$\cL(f)(Z) = \EE[f(Z) - f(Z') | Z]$
	for all polynomials $f \in \cS$.
\end{definition}

Note that this operator is well-defined since for any $f \in \cS$, $\EE[L(f)] = \EE[\EE[f(Z) - f(Z') | Z]] = \EE[f(Z) - f(Z')] = 0$ and hence, $L(f) \in \cS$.

\begin{lemma}\label{lem: eigenvector}
	For all $\al \in \NN^n$, $\chi_{\al}$ is an eigenvector of $\cL$ with eigenvalue $\frac{|\al|_0}{n}$.
\end{lemma}

\begin{proof}
	Recall that $Z'$ is obtained by choosing $i \in [n]$ uniformly at random and then setting $Z' = Z^{(i)}$. Therefore,
    $
		\cL(\chi_{\al})(Z) = \EE[\chi_{\al}(Z) - \chi_{\al}(Z') | Z]
		= \frac{1}{n}\displaystyle\sum_{i \le n} \EE[\chi_{\al}(Z) - \chi_{\al}(Z^{(i)}) | Z]
        $
	When $\al_i = 0$, $\chi_{\al}(Z) - \chi_{\al}(Z^{(i)}) = 0$. Otherwise, $\EE[\chi_{\al}(Z) - \chi_{\al}(Z^{(i)})|Z] = \chi_{\al}(Z)$. Therefore, the above expression simplifies to $\frac{|\al|_0}{n} \chi_{\al}(Z)$.
\end{proof}

\begin{theorem}[Explicit Kernel]\label{thm: explicit_kernel_for_poly}
	For any mean-centered polynomial $f \in \cS$, there exists a polynomial $K_f$ on $2n$ variables $z_1, \ldots, z_n, z_1', \ldots, z_n'$, denoted collectively as $(z, z')$, with the following properties
	\begin{enumerate}
		\item $K_f(z', z) = -K_f(z, z')$
		\item $\EE[K_f(Z, Z') | Z] = f(Z)$ where $(Z, Z')$ is the exchangeable pair we consider above.
	\end{enumerate}
\end{theorem}

\begin{proof}
	Using \cref{propn: basis} and \cref{lem: eigenvector}, under the basis of polynomials $\chi_{\al}$, the operator $\cL$ is a diagonal matrix with nonzero diagonal entries and therefore, $\cL^{-1}$ exists and is explicitly given by
	$\icL(f)(Z) = \displaystyle\sum_{\al} \frac{n}{|\al|_0}\coef{f}{\al} \chi_{\al}(Z)$.
	We then take $K_f(z, z') = \icL(f)(z) - \icL(f)(z')$. The first condition is obvious and for the second condition, we have
	\[\EE[K_f(Z, Z')|Z] = \EE[\icL(f)(Z) - \icL(f)(Z') | Z] = \cL(\icL(f)) = f\]
\end{proof}

As seen in the proof of \cref{thm: explicit_kernel_for_poly}, $\cL$ has a well-defined inverse $\icL$. We now define the matrix $\mK_{k, a, b}$ that we call the \textit{inner kernel}.

\begin{definition}[The inner kernel matrix $\mK_{k, a, b}$]
	For integers $k \ge 1, a, b \ge 0$ such that $a + b < k$, define the matrix $\mK_{k, a, b} \in \RR[Z]^{\cI \times \cK} \times \RR[Z]^{\cJ \times \cK}$ taking $2n$ variables $(z, z') = (z_1, \ldots, z_n, z_1', \ldots, z_n')$ as input as
	$\mK_{k, a, b}(z, z') = \icL(\mG_{k, a, b})(z) - \icL(\mG_{k, a, b})(z')$.
\end{definition}

In the rest of this section except where explicitly stated, fix integers $k \ge 1, a, b \ge 0$ such that $a + b < k$. Then, the inner kernel $\mK_{k, a, b}$ is well-defined.

\begin{lemma}\label{lem: explicit_kernel_for_matrices}
	$\mK_{k, a, b}(Z, Z') = \frac{n}{k - a - b}(\mG_{k, a, b}(Z) - \mG_{k, a, b}(Z'))$
\end{lemma}

\begin{proof}
	\begin{align*}
		\mK_{k, a, b}(Z, Z') &= \cL^{-1}(\mG_{k, a, b})(Z) - \cL^{-1}(\mG_{k, a, b})(Z')\\
		&= \sum_{|\al|_0 = k - a - b} \coef{\mG_{k, a, b}}{\al} (\cL^{-1}(\chi_{\al})(Z) - \cL^{-1}(\chi_{\al})(Z'))\\
		&= \frac{n}{k - a - b}\sum_{|\al|_0 = k - a - b} \coef{\mG_{k, a, b}}{\al} (\chi_{\al}(Z) - \chi_{\al}(Z'))\\
		&= \frac{n}{k - a - b}(\mG_{k, a, b}(Z) - \mG_{k, a, b}(Z'))
	\end{align*}
\end{proof}

The following lemma postulates important properties of the the inner kernel, including how it interacts with $\mD_1$ and $\mD_2$.

\begin{lemma}\label{lem: props_of_exp_kernel_mat}
	$\mK_{k, a, b}$ satisfies the following properties
	\begin{enumerate}
		\item $\mK_{k, a, b}(z', z) = -\mK_{k, a, b}(z, z')$
		\item $\EE[\mK_{k, a, b}(Z, Z') | Z] = \mG_{k, a, b}(Z)$
		\item $(\mD_1(Z) - \mD_1(Z'))\mK_{k, a, b}(Z, Z') = \mK_{k, a, b}(Z, Z')(\mD_2(Z) - \mD_2(Z')) = 0$.
	\end{enumerate}
\end{lemma}

\begin{proof}
	The first equality is obvious from the definition. For the second equality, note that $\EE[\mG_{k, a, b}] = 0$ and $\mK_{k, a, b}$ is defined by replacing each entry $f$ of $\mG_{k, a, b}$ by the kernel polynomial $K_f$ as exhibited in \cref{thm: explicit_kernel_for_poly}. Now, we prove the third equality.

	Consider the matrix $(\mD_1(Z) - \mD_1(Z'))\mK_{k, a, b}(Z, Z')$ whose $[(I, \al_1, \gam_1), (J, \al_2, \gam_2)]$ entry is given by
	\[\frac{n}{k - a - b}\sqrt{\EE[Z^{2\al_1\cdot (1 - \gam_1)}]}(Z^{\al_1 \cdot \gam_1} - (Z')^{\al_1\cdot \gam_1})(\grad_{\al_1 + \al_2}\mX_k[I, J](Z) - \grad_{\al_1 + \al_2}\mX_k[I, J](Z'))\]
	where we have used \cref{lem: explicit_kernel_for_matrices}.
	We will argue that this term is identically $0$.
	We must have $Z' = Z^{(i)}$ for some $i \le n$. If $(\al_1 \cdot \gam_1)_i = 0$, then $Z^{\al_1 \cdot \gam_1} = (Z')^{\al_1\cdot \gam_1}$ and the above term is $0$.
	Otherwise, $(\al_1 + \al_2)_i \neq 0$ and so $\grad_{\al_1 +\al_2}$ on any polynomial $f$ will only contain the terms independent of $Z_i$, in which case $\grad_{\al_1 + \al_2}\mX_k[I, J](Z) = \grad_{\al_1 + \al_2}\mX_k[I, J](Z')$. In this case was well, the above term is $0$. The proof of the other equality is analogous.
\end{proof}

The reason we call $\mK_{k, a, b}$ the inner kernel is because, as seen above, it serves as a kernel for the inner matrix $\mG$ in the decomposition $\mF = \mD\mG\mD$.
Since we will need to work with Hermitian dilations, we define $\mD = \begin{bmatrix}
	\mD_1 & 0\\
	0 & \mD_2
\end{bmatrix}$.
We will use the following basic fact extensively in our manipulations.

\begin{fact}
	For any matrix $\mA \in \RR[Z]^{\cI\times \cK} \times \RR[Z]^{\cJ\times \cK}$, $\mD \herm{\mA}\mD = \herm{\mD_1\mA\mD_2}$.
\end{fact}

\begin{proof}
	We have
	\begin{align*}
		\mD \herm{\mA}\mD =
		\begin{bmatrix}
			\mD_1 & 0\\
			0 & \mD_2
		\end{bmatrix}
		\begin{bmatrix}
			0 & \mA \\
			\mA^\T & 0
		\end{bmatrix}
		\begin{bmatrix}
			\mD_1 & 0\\
			0 & \mD_2
		\end{bmatrix}
		&=
		\begin{bmatrix}
			0 & \mD_1\mA\\
			\mD_2\mA^\T  & 0
		\end{bmatrix}
		\begin{bmatrix}
			\mD_1 & 0\\
			0 & \mD_2
		\end{bmatrix}\\
		&=
		\begin{bmatrix}
			0 & \mD_1\mA\mD_2\\
			\mD_2\mA^\T\mD_1  & 0
		\end{bmatrix}\\
		&= \herm{\mD_1\mA\mD_2}
	\end{align*}
\end{proof}

We start with a generalized version of a result from \cite{paulin2016}.

\begin{lemma}\label{lem: deviation_bound}
	Let $\mK = \herm{\mK}_{k, a, b}$. For any symmetric matrix valued function $\mR$ on the variables $Z$ of the same dimensions as $\mK$, such that $\EE\norm{\mK(Z, Z')\mR(Z)}  < \infty$, we have
	\[\EE[\herm{\mF}_{k, a, b}(Z)\mR(Z)] = \frac{1}{2}\EE[\mD(Z)\mK(Z, Z')\mD(Z) (\mR(Z) - \mR(Z'))]\]
\end{lemma}

\begin{proof}
	By \cref{lem: props_of_exp_kernel_mat}, we have
	\begin{align*}
		\EE[\herm{\mF}_{k, a, b}(Z)\mR(Z)] &= \EE[\mD(Z)\herm{\mG}_{k, a, b}(Z)\mD(Z) \mR(Z)]\\
		&= \EE[\mD(Z)\EE[\mK(Z, Z') | Z]\mD(Z) \mR(Z)]\\
		&= \EE[\mD(Z)\mK(Z, Z')\mD(Z) \mR(Z)]
	\end{align*}
	where the first equality follow from condition $2$ of \cref{lem: props_of_exp_kernel_mat} and the second follows from the pull-through property of expectations. Continuing,
	\begin{align*}
		\EE[\herm{\mF}_{k, a, b}(Z)\mR(Z)] &= \EE[\mD(Z)\mK(Z, Z')\mD(Z)\mR(Z)]\\
		&= \EE[\mD(Z')\mK(Z', Z)\mD(Z') \mR(Z')]\\
		&= -\EE[\mD(Z')\mK(Z, Z')\mD(Z') \mR(Z')]\\
		&= -\EE[\mD(Z)\mK(Z, Z')\mD(Z') \mR(Z')]\\
		&= -\EE[\mD(Z)\mK(Z, Z')\mD(Z) \mR(Z')]
	\end{align*}
	Here, the second equality follows from the fact that $(Z, Z')$ has the same distribution as $(Z', Z)$, so we can exchange them. The third, fourth and fifth equalities follow from conditions $1, 3, 3$ of \cref{lem: props_of_exp_kernel_mat} respectively. Adding the two displays, we get the result.
\end{proof}

\begin{definition}[Matrices $\mU_{k, a, b}, \mV_{k, a, b}$]
	We define the following matrices
	\[\mU_{k, a, b} = \EE[(\herm{\mF}_{k, a, b}(Z) - \herm{\mF}_{k, a, b}(Z'))^2|Z]\]
	\[\mV_{k, a, b} = \EE[(\mD(Z)\herm{\mK}_{k, a, b}(Z, Z')\mD(Z))^2|Z]\]
\end{definition}

The definition of $\mU_{k, a, b}$ is essentially unchanged from \cite{paulin2016}, where it is called the \textit{conditional variance}. The definition of $\mV_{k, a, b}$ is slightly different in our setting. This lets us exploit the specific product structure exhibited by $\herm{\mF}_{k, a, b}$ and the special properties of the inner kernel from \cref{lem: props_of_exp_kernel_mat}.
We will now prove a lemma which is similar to a lemma shown in \cite{paulin2016}.

\begin{lemma}\label{lem: main_pmt_bound}
	For any $s > 0$ and for any integer $t \ge 1$,
	\begin{align*}
		\Esch{\herm{\mF}_{k, a, b}}{2t}
		&\le \left(\frac{2t - 1}{4}\right)^t\Esch{s\mU_{k, a, b} + s^{-1}\mV_{k, a, b}}{t}
	\end{align*}
\end{lemma}

To prove this, we will use the following inequality.

\begin{lemma}[Polynomial mean value trace inequality, \cite{paulin2016}]\label{lem: mean_value_trace_inequality}
	For all matrices $\mA, \mB, \mC \in \HH^d$, all integers $q \ge 1$ and all $s > 0$,
	\begin{align*}
		\tr [\mC(\mA^q - \mB^q)]| \le \frac{q}{4} \tr[(s(\mA - \mB)^2 + s^{-1}\mC^2)(\mA^{q - 1} + \mB^{q - 1})]
	\end{align*}
\end{lemma}

\begin{proof}[Proof of \cref{lem: main_pmt_bound}]
	We start by invoking \cref{lem: deviation_bound} by setting $\mR(Z) = \herm{\mF}_{k, a, b}^{2t - 1}(Z)$ to get
	\begin{align*}
		\Esch{\herm{\mF}_{k, a, b}}{2t} &= \EE\tr [\herm{\mF}_{k, a, b}\cdot\herm{\mF}_{k, a, b}^{2t - 1}]
		= \frac{1}{2}\EE[\mD(Z)\herm{\mK}_{k, a, b}(Z, Z')\mD(Z) (\herm{\mF}_{k, a, b}^{2t - 1}(Z) - \herm{\mF}_{k, a, b}^{2t - 1}(Z'))]
	\end{align*}

	Applying \cref{lem: mean_value_trace_inequality},
	\begin{align*}
		&\Esch{\herm{\mF}_{k, a, b}}{2t}\\
		&\le (\frac{2t - 1}{8})\EE\tr[(s(\herm{\mF}_{k, a, b}(Z) - \herm{\mF}_{k, a, b}(Z'))^2 + s^{-1}(\mD(Z)\herm{\mK}_{k, a, b}(Z, Z')\mD(Z))^2)(\herm{\mF}_{k, a, b}^{2t - 2}(Z) + \herm{\mF}_{k, a, b}^{2t - 2}(Z'))]\\
		&= (\frac{2t - 1}{4})\EE\tr[(s(\herm{\mF}_{k, a, b}(Z) - \herm{\mF}_{k, a, b}(Z'))^2 + s^{-1}(\mD(Z)\herm{\mK}_{k, a, b}(Z, Z')\mD(Z))^2)\herm{\mF}_{k, a, b}^{2t - 2}(Z)]
	\end{align*}
	where the last line used the fact that $(Z, Z')$ has the same distribution as $(Z', Z)$ and applied condition $3$ of \cref{lem: props_of_exp_kernel_mat}. Using the definitions of $\mU_{k, a, b}$ and $\mV_{k, a, b}$, we get
	\begin{align*}
		\Esch{\herm{\mF}_{k, a, b}}{2t} &\le \frac{2t - 1}{4}\EE\tr[(s\mU_{k, a, b} + s^{-1}\mV_{k, a, b})\herm{\mF}_{k, a, b}^{2t - 2}]\\
		&\le \frac{2t - 1}{4}\left(\Esch{s\mU_{k, a, b} + s^{-1}\mV_{k, a, b}}{t}\right)^{1/t}(\Esch{\herm{\mF}_{k, a, b}}{2t})^{(t - 1)/t}
	\end{align*}
	where we used H\"{o}lder's inequality for the trace and H\"{o}lder's inequality for the expectation. Rearranging gives the result.
\end{proof}

\subsection{Proof of \cref{lem: main_general}}

\cref{lem: main_pmt_bound} suggests that in order to bound $\Esch{\herm{\mF}_{k, a, b}}{2t}$, it suffices to bound $\Esch{\mU_{k, a, b}}{t}$ and $\Esch{\mV_{k, a, b}}{t}$. Indeed, this will be our strategy. To bound $\Esch{\mU_{k, a, b}}{t}$, we will bound it via the matrices that we define below.

\begin{definition}[Matrices $\mDel_1^{k, a, b}, \mDel_2^{k, a, b}, \mDel_3^{k, a, b}$]
	Define the matrices
	\[\mDel_1^{k, a, b} = \EE[((\mD(Z) - \mD(Z'))\herm{\mG}_{k, a, b}(Z)\mD(Z))^2|Z]\]
	\[\mDel_2^{k, a, b} = \EE[(\mD(Z)(\herm{\mG}_{k, a, b}(Z) - \herm{\mG}_{k, a, b}(Z'))\mD(Z))^2|Z]\]
	\[\mDel_3^{k, a, b} = \EE[(\mD(Z)\herm{\mG}_{k, a, b}(Z)(\mD(Z) - \mD(Z')))^2|Z]\]
\end{definition}

\begin{lemma}\label{lem: bound_U_by_Deltas}
	$\mU_{k, a, b} \preceq 3(\mDel_1^{k, a, b} + \mDel_2^{k, a, b} + \mDel_3^{k, a, b})$.
\end{lemma}

To prove this lemma, we will use the following lemma.

\begin{lemma}\label{lem: orthogonality}
	We have the relations
	\[(\mD(Z) - \mD(Z'))(\herm{\mG}_{k, a, b}(Z)\mD(Z) - \herm{\mG}_{k, a, b}(Z')\mD(Z')) = 0\]
	\[(\herm{\mG}_{k, a, b}(Z) - \herm{\mG}_{k, a, b}(Z'))(\mD(Z) - \mD(Z')) = 0\]
\end{lemma}

\begin{proof}[Proof sketch]
	The proof is similar to the proof of third equality in \cref{lem: props_of_exp_kernel_mat}. When $Z'$ is set to $Z^{(i)}$ for some $i \le n$, when a diagonal entry of $\mD(Z) - \mD(Z')$ is nonzero, then the corresponding row of $\herm{\mG}_{k, a, b}(Z)\mD(Z) - \herm{\mG}_{k, a, b}(Z')\mD(Z')$ will be $0$. The second equality is analogous.
\end{proof}

\begin{proof}[Proof of \cref{lem: bound_U_by_Deltas}]
	We have
	{\footnotesize
	\begin{align*}
		&(\herm{\mF}_{k, a, b}(Z) - \herm{\mF}_{k, a, b}(Z'))^2\\
		&= (\mD(Z)\herm{\mG}_{k, a, b}(Z)\mD(Z) - \mD(Z')\herm{\mG}_{k, a, b}(Z')\mD(Z'))^2\\
		&= \bigg(\mD(Z)\herm{\mG}_{k, a, b}(Z)(\mD(Z) - \mD(Z')) + \mD(Z)(\herm{\mG}_{k, a, b}(Z) - \herm{\mG}_{k, a, b}(Z'))\mD(Z') + (\mD(Z) - \mD(Z'))\herm{\mG}_{k, a, b}(Z')\mD(Z')\bigg)^2\\
		&= \bigg(\mD(Z)\herm{\mG}_{k, a, b}(Z)(\mD(Z) - \mD(Z')) + \mD(Z)(\herm{\mG}_{k, a, b}(Z) - \herm{\mG}_{k, a, b}(Z'))\mD(Z) + (\mD(Z) - \mD(Z'))\herm{\mG}_{k, a, b}(Z)\mD(Z)\bigg)^2
	\end{align*}
	}
	where the last equality follows from \cref{lem: orthogonality}. Taking expectations conditioned on $Z$ and applying \cref{fact: cs}, we immediately get $\mU_{k, a, b} \preceq 3(\mDel_1^{k, a, b} + \mDel_2^{k, a, b} + \mDel_3^{k, a, b})$.
\end{proof}

In subsequent sections, we will prove the following technical bounds on the matrices we have considered so far.

\begin{restatable}{lemma}{boundDelTwo}\label{lem: bound_on_Del2}
	For all integers $t \ge 1$, $\Esch{\mDel_2^{k, a, b}}{t} \le \frac{(2\dpoly)^t}{n^t} (\Esch{\herm{\mF}_{k, a, b + 1}}{2t} + \Esch{\herm{\mF}_{k, a + 1, b}}{2t})$.
\end{restatable}

\begin{restatable}{lemma}{boundV}\label{lem: bound_on_V}
	For all integers $t \ge 1$, $\Esch{\mV_{k, a, b}}{t} \le (2\dpoly)^tn^t (\Esch{\herm{\mF}_{k, a, b + 1}}{2t} + \Esch{\herm{\mF}_{k, a + 1, b}}{2t})$.
\end{restatable}

\begin{restatable}{lemma}{boundDelOne}\label{lem: bound_on_Del1}
	For all integers $t \ge 1$, $\Esch{\mDel_1^{k, a, b}}{t} \le \frac{(8d\dpoly)^t}{n^t}\Esch{\herm{\mF}_{k, a, b}}{2t}$.
\end{restatable}

\begin{restatable}{lemma}{boundDelThree}\label{lem: bound_on_Del3}
	For all integers $t \ge 1$, $\Esch{\mDel_3^{k, a, b}}{t} \le \frac{(4\dpoly)^t}{n^t}\Esch{\herm{\mF}_{k, a, b}}{2t}$.
\end{restatable}

Assuming the above lemmas, we can complete the proof of \cref{lem: main_general}, which we restate for convenience.

\maingeneral*

\begin{proof}[Proof of \cref{lem: main_general}]
	Using \cref{lem: main_pmt_bound}, \cref{lem: bound_U_by_Deltas}, we get that for any $s > 0$,
	\begin{align*}
		\Esch{\herm{\mF}_{k, a, b}}{2t}
		&\le (\frac{2t - 1}{4})^t\Esch{s\mU_{k, a, b} + s^{-1}\mV_{k, a, b}}{t}\\
		&\le t^t(s^t\Esch{\mU_{k, a, b}}{t} + s^{-t}\Esch{\mV_{k, a, b}}{t})\\
		&\le (9st)^t(\Esch{\mDel_1^{k, a, b}}{t} + \Esch{\mDel_2^{k, a, b}}{t} + \Esch{\mDel_3^{k, a, b}}{t}) + t^ts^{-t}\Esch{\mV_{k, a, b}}{t}
	\end{align*}
	Let $\rho = s / n$. Since the inequality is true for any choice of $s > 0$, it is true for any choice of $\rho > 0$.
	Now, using \cref{lem: bound_on_Del1}, \cref{lem: bound_on_Del3},
	\begin{align*}
		(9st)^t(\Esch{\mDel_1^{k, a, b}}{t} + \Esch{\mDel_3^{k, a, b}}{t}) &\le (9st)^t\bigg(\frac{(8d\dpoly)^t}{n^t} + \frac{(4\dpoly)^t}{n^t}\bigg)\Esch{\herm{\mF}_{k, a, b}}{2t}\\
		&= \rho^t (C_1td\dpoly)^t\Esch{\herm{\mF}_{k, a, b}}{2t}
	\end{align*}
	for an absolute constant $C_1 > 0$. Using \cref{lem: bound_on_Del2}, \cref{lem: bound_on_V},
	\begin{align*}
		(9st)^t\Esch{\mDel_2^{k, a, b}}{t} + t^ts^{-t}\Esch{\mV_{k, a, b}}{t} & \le\bigg((9st)^t\frac{(2\dpoly)^t}{n^t} + t^ts^{-t}(2\dpoly)^tn^t\bigg)(\Esch{\herm{\mF}_{k, a, b + 1}}{2t} + \Esch{\herm{\mF}_{k, a + 1, b}}{2t})\\
		&\le (\rho^tC_2^t + \rho^{-t}C_3^t) (t\dpoly)^t (\Esch{\herm{\mF}_{k, a, b + 1}}{2t} + \Esch{\herm{\mF}_{k, a + 1, b}}{2t})
	\end{align*}
	for absolute constants $C_2, C_3 > 0$.
	Therefore,
	\begin{align*}
		\Esch{\herm{\mF}_{k, a, b}}{2t} &\le \rho^t (C_1td\dpoly)^t\Esch{\herm{\mF}_{k, a, b}}{2t} + (\rho^tC_2^t + \rho^{-t}C_3^t) (t\dpoly)^t(\Esch{\herm{\mF}_{k, a, b + 1}}{2t} + \Esch{\herm{\mF}_{k, a + 1, b}}{2t})
	\end{align*}
	We choose $\rho > 0$ so that $\rho^t (C_1td\dpoly)^t = \frac{1}{2}$ to get
	\begin{align*}
		\Esch{\herm{\mF}_{k, a, b}}{2t} &\le \frac{1}{2}\Esch{\herm{\mF}_{k, a, b}}{2t} + \frac{1}{2}(Ct^2d\dpoly^2)^t (\Esch{\herm{\mF}_{k, a, b + 1}}{2t} + \Esch{\herm{\mF}_{k, a + 1, b}}{2t})
	\end{align*}
	for an absolute constant $C > 0$.
	Rearranging yields the result.
\end{proof}

\subsection{Bounding $\mDel_2^{k, a, b}$ and $\mV_{k, a, b}$}

The next lemma relates $\mV_{k, a, b}$ to $\mDel_2^{k, a, b}$ upto a factor of $n^2$ which will be enough for us. We can then focus on bounding $\mDel_2^{k, a, b}$.

\begin{lemma}\label{lem: bounding_V_loewner}
	$\mV_{k, a, b} \preceq n^2 \mDel_2^{k, a, b}$
\end{lemma}

\begin{proof}
	Using \cref{lem: explicit_kernel_for_matrices},
	\begin{align*}
		\mV_{k, a, b} &= \EE[(\mD(Z)\herm{\mK}_{k, a, b}(Z, Z')\mD(Z))^2|Z]\\
		&= \EE[(\mD(Z)\bigg(\frac{n}{k - a - b}(\herm{\mG}_{k, a, b}(Z) - \herm{\mG}_{k, a, b}(Z'))\bigg)\mD(Z))^2|Z]\\
		&\preceq n^2\EE[(\mD(Z)(\herm{\mG}_{k, a, b}(Z) - \herm{\mG}_{k, a, b}(Z'))\mD(Z))^2|Z]\\
		&= n^2 \mDel_2^{k, a, b}
	\end{align*}
\end{proof}

For $1 \le i \le n$ and $1 \le l \le d$, let $\mat{e}_{i, l} \in \NN^n$ denote the vector $\al$ with $\al_i = l$ and $\al_j = 0$ for $j \neq i$.
We note the following simple proposition.

\begin{propn}\label{propn: difference_equality}
	For any polynomial $f$ such that the degree of $Z_i$ is at most $d$, $f(Z) - f(Z^{(i)}) = \displaystyle\sum_{1 \le l \le d} (Z_i^l - \resamp{Z_i}^l)\grad_{\mat{e}_{i, l}}(f)$
\end{propn}

We now restate and prove \cref{lem: bound_on_Del2}.

\boundDelTwo*

\begin{proof}
	Consider
	\begin{align*}
		\mDel_2^{k, a, b} &= \EE[(\mD(Z)(\herm{\mG}_{k, a, b}(Z) - \herm{\mG}_{k, a, b}(Z'))\mD(Z))^2|Z]\\
		&= \EE\bigg[ \begin{bmatrix}
			\mM\mM^\T & 0\\
			0 & \mM^\T\mM
		\end{bmatrix} | Z\bigg]\\
		&= \begin{bmatrix}
			\EE[\mM\mM^\T|Z] & 0\\
			0 & \EE[\mM^\T\mM|Z]
		\end{bmatrix}
	\end{align*}
	where $\mM = \mD_1(Z)(\mG_{k, a, b}(Z) - \mG_{k, a, b}(Z'))\mD_2(Z)$. Using \cref{propn: difference_equality},
	\begin{align*}
		\EE[\mM\mM^T | Z] &= \EE[\mD_1(Z)(\mG_{k, a, b}(Z) - \mG_{k, a, b}(Z'))\mD_2(Z)\cdot \mD_2(Z) (\mG_{k, a, b}(Z) - \mG_{k, a, b}(Z'))^\T\mD_1(Z)|Z]\\
		&= \frac{1}{n} \sum_{i = 1}^n\EE[\mD_1(Z)(\mG_{k, a, b}(Z) - \mG_{k, a, b}(Z^{(i)}))\mD_2(Z)\cdot \mD_2(Z) (\mG_{k, a, b}(Z) - \mG_{k, a, b}(Z^{(i)}))^\T\mD_1(Z)|Z]\\
		&= \frac{1}{n}\sum_{i = 1}^n \sum_{l = 1}^d\EE[(Z_i^l - \resamp{Z_i}^l)^2|Z]\cdot \mD_1(Z)(\grad_{\mat{e}_{i, l}} \mG_{k, a, b})(Z)\mD_2(Z)\cdot \mD_2(Z) (\grad_{\mat{e}_{i, l}} \mG_{k, a, b})(Z)^\T\mD_1(Z)
	\end{align*}
	Define $\mN_{i, l}(Z) := \mD_1(Z)(\grad_{\mat{e}_{i, l}} \mG_{k, a, b})(Z)\mD_2(Z)$. Then,
	\begin{align*}
		\EE[\mM\mM^T | Z] = \frac{1}{n}\sum_{i = 1}^n \sum_{l = 1}^d\EE[(Z_i^l - \resamp{Z_i}^l)^2|Z]\cdot \mN_{i, l}(Z)\mN_{i, l}(Z)^\T
		\preceq \frac{2}{n}\sum_{i = 1}^n \sum_{l = 1}^d(Z_i^{2l} + \EE[Z_i^{2l}])\cdot \mN_{i, l}(Z)\mN_{i, l}(Z)^\T
	\end{align*}
	Similarly,
    $
		\EE[\mM^\T\mM | Z] \preceq \frac{2}{n}\displaystyle\sum_{i = 1}^n \sum_{l = 1}^d(Z_i^{2l} + \EE[Z_i^{2l}])\cdot \mN_{i, l}(Z)^\T\mN_{i, l}(Z)
        $

	\begin{claim}\label{claim: reduction_general}
		We have the relations
		\[\sum_{i = 1}^n\sum_{l = 1}^d (Z_i^{2l} + \EE[Z_i^{2l}]) \cdot \mN_{i, l}(Z)\mN_{i, l}(Z)^\T = (b + 1)\mF_{k, a, b + 1}\mF_{k, a, b + 1}^\T\]
		\[\sum_{i = 1}^n \sum_{l = 1}^d(Z_i^{2l} + \EE[Z_i^{2l}])\cdot \mN_{i, l}(Z)^\T\mN_{i, l}(Z) = (a + 1)\mF_{k, a + 1, b}^\T\mF_{k, a + 1, b}\]
	\end{claim}

	Using this claim, we have
	\[\EE[\mM\mM^T | Z] \preceq \frac{2(b + 1)}{n}\mF_{k, a, b + 1}\mF_{k, a, b + 1}^\T \preceq \frac{2\dpoly}{n}\mF_{k, a, b + 1}\mF_{k, a, b + 1}^\T\]
	\[\EE[\mM^\T\mM | Z] \preceq \frac{2(a + 1)}{n}\mF_{k, a + 1, b}^\T\mF_{k, a + 1, b} \preceq \frac{2\dpoly}{n}\mF_{k, a + 1, b}^\T\mF_{k, a + 1, b}\]
	Therefore,
	\begin{align*}
		\Esch{\mDel_2^{k, a, b}}{t} = \Esch{\EE[\mM\mM^\T|Z]}{t} + \Esch{\EE[\mM^\T\mM|Z]}{t}
		&\le \frac{(2\dpoly)^t}{n^t} (\Esch{\mF_{k, a, b + 1}}{2t} + \Esch{\mF_{k, a + 1, b}}{2t})\\
		&\le \frac{(2\dpoly)^t}{n^t} (\Esch{\herm{\mF}_{k, a, b + 1}}{2t} + \Esch{\herm{\mF}_{k, a + 1, b}}{2t})
	\end{align*}
\end{proof}

It remains to prove the claim.
\begin{proof}[Proof of~\cref{claim: reduction_general}]
	We will prove the first relation, the second is analogous.
	For a fixed $i \le n, l \le d$, consider any nonzero entry $[(I_1, \al_1, \gam_1), (I_2, \al_2, \gam_2)]$ of $\sum_{i = 1}^n \sum_{l = 1}^d (Z_i^{2l} + \EE[Z_i^{2l}]) \mN_{i, l}(Z) \mN_{i, l}(Z)^\T$, where $I_1, I_2 \in \cI, (\al_1, \gam_1), (\al_2, \gam_2) \in \cK$. We must have $|\al_1|_0 = |\al_2|_0 = a$, in which case the entry  is equal to
	\begin{align*}
		\sum_{\substack{(J, \al_3, \gam_3) \in \cJ\times \cK\\ |\al_3| = b \\ \al_1\al_3 = \al_2\al_3 = 0}} &(Z_i^{2l} + \EE[Z_i^{2l}]) \cdot (\sqrt{\EE[Z^{2\al_1\cdot (1 - \gam_1) + 2\al_3\cdot (1 - \gam_3)}]}Z^{\al_1\cdot \gam_1 + \al_3\cdot\gam_3}\grad_{\mat{e}_{i, l}} \grad_{\al_1 + \al_3} \mX_k[I_1, J])\\
		&\cdot (\sqrt{\EE[Z^{2\al_2\cdot (1 - \gam_2) + 2\al_3\cdot (1 - \gam_3)}]}Z^{\al_2\cdot \gam_2 + \al_3\cdot\gam_3}\grad_{\mat{e}_{i, l}} \grad_{\al_2 + \al_3} \mX_k[I_2, J])
	\end{align*}
	Note that the term inside the summation is nonzero only when $\mat{e}_{i, l}\cdot (\al_1 + \al_3) = \mat{e}_{i, l} \cdot (\al_2 + \al_3) = 0$. Hence, this sum can be written as
	\begin{align*}
		\sum_{\substack{(J, \al_3, \gam_3) \in \cJ\times \cK\\ |\al_3| = b + 1 \\ \mat{e}_{i, l} \unlhd \al_3, \al_1\al_3 = \al_2\al_3 = 0}} &(\sqrt{\EE[Z^{2\al_1\cdot (1 - \gam_1) + 2\al_3\cdot (1 - \gam_3)}]}Z^{\al_1\cdot \gam_1 + \al_3\cdot\gam_3}\grad_{\al_1 + \al_3} \mX_k[I_1, J])\\
		&\cdot (\sqrt{\EE[Z^{2\al_2\cdot (1 - \gam_2) + 2\al_3\cdot (1 - \gam_3)}]}Z^{\al_2\cdot \gam_2 + \al_3\cdot\gam_3}\grad_{\al_2 + \al_3} \mX_k[I_2, J])
	\end{align*}
	When we add this entry over all $i \le n, l \le d$, this simplifies to
	\begin{align*}
		(b + 1) \cdot \sum_{\substack{(J, \al_3, \gam_3) \in \cJ\times \cK\\ |\al_3| = b + 1 \\ \al_1\al_3 = \al_2\al_3 = 0}} &(\sqrt{\EE[Z^{2\al_1\cdot (1 - \gam_1) + 2\al_3\cdot (1 - \gam_3)}]}Z^{\al_1\cdot \gam_1 + \al_3\cdot\gam_3}\grad_{\al_1 + \al_3} \mX_k[I_1, J])\\
		&\cdot (\sqrt{\EE[Z^{2\al_2\cdot (1 - \gam_2) + 2\al_3\cdot (1 - \gam_3)}]}Z^{\al_2\cdot \gam_2 + \al_3\cdot\gam_3}\grad_{\al_2 + \al_3} \mX_k[I_2, J])
	\end{align*}
	The factor of $(b + 1)$ came because the index $i$ could have been chosen from among all the active indices in $\al_3$. But this is precisely the $[(I_1, \al_1, \gam_1), (I_2, \al_2, \gam_2)]$ entry of $(b + 1)\mF_{k, a, b + 1}\mF_{k, a, b + 1}^\T$, proving the claim.
\end{proof}

We restate and prove \cref{lem: bound_on_V}.

\boundV*

\begin{proof}
	Using \cref{lem: bounding_V_loewner} and \cref{lem: bound_on_Del2}, we get
	\begin{align*}
		\Esch{\mV_{k, a, b}}{t} \le n^{2t}\Esch{\mDel_2^{k, a, b}}{t}
		\le (2\dpoly)^tn^t (\Esch{\herm{\mF}_{k, a, b + 1}}{2t} + \Esch{\herm{\mF}_{k, a + 1, b}}{2t})
	\end{align*}
\end{proof}

\subsection{Bounding $\mDel_1^{k, a, b}$ and $\mDel_3^{k, a, b}$}

Define $\sqcup$ to be the disjoint union of sets. For $1 \le i \le n$ and $1 \le l \le d$, define the diagonal matrices $\mPi_{i, l}, \mPi_{i, l}', \mPi_i, \mPi_i' \in \RR^{(\cI \times \cK) \sqcup (\cJ \times \cK)} \times \RR^{(\cI \times \cK) \sqcup (\cJ \times \cK)}$ (the same dimensions as $\mD$) as
\[\mPi_{i, l}[(I, \al, \beta), (I, \al, \beta)] = \begin{dcases}
	1 & \text{ if $(\al \cdot \gam)_i \neq 0$ and $\al_i = l$}\\
	0 & \text{ o.w.}
\end{dcases}\qquad \mPi_i[(I, \al, \beta), (I, \al, \beta)] = \begin{dcases}
	1 & \text{ if $(\al \cdot \gam)_i \neq 0$}\\
	0 & \text{ o.w.}
\end{dcases}\]
\[\mPi'_{i, l}[(I, \al, \beta), (I, \al, \beta)] = \begin{dcases}
	1 & \text{ if $\al_i \neq 0$ and $\al_i = l$}\\
	0 & \text{ o.w.}
\end{dcases}\qquad \mPi_i'[(I, \al, \beta), (I, \al, \beta)] = \begin{dcases}
	1 & \text{ if $\al_i \neq 0$}\\
	0 & \text{ o.w.}
\end{dcases}\]
for all $I \in \cI \sqcup \cJ$.
Note that for all $i \le n$, $\mPi_i = \sum_{l = 1}^d \mPi_{i, l}$.

Also, for all $1 \le i \le n$, we define the permutation matrices $\mSig_i \in\RR^{(\cI \times \cK) \sqcup (\cJ \times \cK)} \times \RR^{(\cI \times \cK) \sqcup (\cJ \times \cK)}$ as follows. Consider the permutation $\sig_1$ on $\cI\times \cK$ that transposes $(I, \al, \gam)$ and $(I, \al, \gam + \mat{e}_i)$ for all $(I, \al, \gam) \in \cI\times \cK$ such that $\al_i \neq 0$. Here, $\mat{e}_i \in \{0, 1\}^n$ has exactly one nonzero entry, which is in the $i$th position, and $\gam + \mat{e}_i$ is the usual addition over $\mathbb{F}_2$. $\sig_1$ leaves other positions fixed. Let $\mSig^{(1)}_i$ be the permutation matrix for $\sig$. Similarly, let $\mSig^{(2)}_i$ be the permutation matrix of the permutation $\sig_2$ on $\cJ \times \cK$ that transposes $(J, \al, \gam)$ and $(J, \al, \gam + \mat{e}_i)$ for all $(J, \al, \gam) \in \cJ\times \cK$ such that $\al_i \neq 0$, and leaves all other positions fixed. Then, we define $\mSig_i = \begin{bmatrix}
	\mSig^{(1)}_i & 0\\
	0 & \mSig^{(2)}_i
\end{bmatrix}$. The following fact is easy to verify.

\begin{fact}\label{fact: commutativity}
	$\mPi'_{i, l}\mSig_i = \mSig_i\mPi'_{i, l}$ and $\mPi_i' \mSig_i = \mSig_i \mPi_i'$.
\end{fact}

We are now ready to prove \cref{lem: bound_on_Del1} which we restate for convenience.

\boundDelOne*

\begin{proof}
	Firstly,
	\begin{align*}
		\mDel_1^{k, a, b} &= \EE[((\mD(Z) - \mD(Z'))\herm{\mG}_{k, a, b}(Z)\mD(Z))^2|Z]\\
		&= \EE[(\mD(Z) - \mD(Z'))\herm{\mG}_{k, a, b}(Z)\mD(Z)\cdot \mD(Z)\herm{\mG}_{k, a, b}(Z)(\mD(Z) - \mD(Z'))|Z]\\
		&= \EE[(\mD(Z) - \mD(Z')) \mM(Z) (\mD(Z) - \mD(Z'))|Z]
	\end{align*}
	where we define $\mM(Z) = \herm{\mG}_{k, a, b}(Z)\mD(Z)\cdot \mD(Z)\herm{\mG}_{k, a, b}(Z)$.
	Recall that $Z' = Z^{(i)}$ for some $i$ randomly chosen from $[n]$ uniformly. Observing that $\mD(Z) - \mD(Z^{(i)}) = \mPi_i(\mD(Z) - \mD(Z^{(i)}))$ for all $i$, we get
	\begin{align*}
		\mDel_1^{k, a, b}
		&= \EE[ \EE_{i \in [n]} [(\mD(Z) - \mD(Z^{(i)})) \mM(Z) (\mD(Z) - \mD(Z^{(i)}))]|Z]\\
		&= \EE[\EE_{i \in [n]} [\mPi_i(\mD(Z) - \mD(Z^{(i)})) \mM(Z) (\mD(Z) - \mD(Z^{(i)}))\mPi_i]|Z]\\
		&\preceq 2\bigg(\EE[\EE_{i \in [n]} [\mPi_i\mD(Z)\mM(Z)\mD(Z)\mPi_i]|Z] + \EE[\EE_{i \in [n]} [\mPi_i\mD(Z^{(i)})\mM(Z)\mD(Z^{(i)})\mPi_i]|Z]\bigg)\\
		&\preceq 2\bigg(\EE_{i \in [n]} [\mPi_i\herm{\mF}_{k, a, b}^2\mPi_i] + \EE[\EE_{i \in [n]} [\mPi_i\mD(Z^{(i)})\mM(Z)\mD(Z^{(i)})\mPi_i]|Z]\bigg)\\
		& \preceq 2(\mDel_{10} + \mDel_{11})
	\end{align*}
	where we define
	\[\mDel_{10} = \EE_{i \in [n]} [\mPi_i\herm{\mF}_{k, a, b}^2\mPi_i], \qquad \mDel_{11} = \EE[\EE_{i \in [n]} [\mPi_i\mD(Z^{(i)})\mM(Z)\mD(Z^{(i)})\mPi_i]|Z]\]
	Invoking \cref{lem: jensen_trace} over the interval $[0, \infty)$ with the convex continuous function $f(x) = x^t$, $\mB_i = \herm{\mF}_{k, a, b}^2, \mA_i = \frac{1}{\sqrt{\dpoly}}\mPi_i$ where we observe that $\sum_{i = 1}^n \mA_i \mA_i^T = \frac{1}{\dpoly}\sum_{i = 1}^n \mPi_i^2\preceq \mI$, we get
	\begin{align*}
		\Esch{\mDel_{10}}{t} = \EE\tr[\mDel_{10}^t] = \EE\tr[\bigg(\EE_{i \in [n]} [\mPi_i\herm{\mF}_{k, a, b}^2\mPi_i]\bigg)^t] &= \frac{1}{n^t}\EE\tr[\bigg(\sum_{i = 1}^n\mPi_i\herm{\mF}_{k, a, b}^2\mPi_i\bigg)^t]\\
		&\le \frac{\dpoly^{t - 1}}{n^t}\EE\tr[\bigg(\sum_{i = 1}^n\mPi_i\herm{\mF}_{k, a, b}^{2t}\mPi_i\bigg)]\\
		&\le \frac{\dpoly^{t - 1}}{n^t}\EE\tr[\bigg(\sum_{i = 1}^n\mPi_i^2\bigg)\herm{\mF}_{k, a, b}^{2t}]\\
		&\le \frac{\dpoly^t}{n^t}\EE\tr[\herm{\mF}_{k, a, b}^{2t}]\\
		&= \frac{\dpoly^t}{n^t}\Esch{\herm{\mF}_{k, a, b}}{2t}
	\end{align*}

	Now, consider
	\begin{align*}
		\mDel_{11} &= \EE[\EE_{i \in [n]} [\mPi_i\mD(Z^{(i)})\mM(Z)\mD(Z^{(i)})\mPi_i]|Z]\\
		&= \EE[\EE_{i \in [n]} 	[(\sum_{l = 1}^d\mPi_{i, l})\mD(Z^{(i)})\mM(Z)\mD(Z^{(i)})(\sum_{l = 1}^d\mPi_{i, l})]|Z]\\
		&\preceq d\cdot \EE[\EE_{i \in [n]} [\sum_{l = 1}^d\mPi_{i, l}\mD(Z^{(i)})\mM(Z)\mD(Z^{(i)})\mPi_{i, l}]|Z]\\
		&= d\cdot \EE_{i \in [n]} [\sum_{l = 1}^d \frac{\EE[Z_i^{2l}]}{Z_i^{2l}}\mPi_{i, l}\mD(Z)\mM(Z)\mD(Z)\mPi_{i, l}]\\
		&= \frac{d}{n} \sum_{i = 1}^n\sum_{l = 1}^d \frac{\EE[Z_i^{2l}]}{Z_i^{2l}}\mPi_{i, l}\mD(Z)\mM(Z)\mD(Z)\mPi_{i, l}\\
		&= \frac{d}{n} \sum_{i = 1}^n\sum_{l = 1}^d \mPi_{i, l}\mSig_i\mD(Z)\mM(Z)\mD(Z)\mSig_i^\T\mPi_{i, l}\\
		&= \frac{d}{n} \sum_{i = 1}^n\sum_{l = 1}^d \mPi_{i, l}\mSig_i\herm{\mF}_{k, a, b}^2\mSig_i^\T\mPi_{i, l}
	\end{align*}
	We now invoke \cref{lem: jensen_trace} on $d\dpoly$ terms with $\mB_{i, l} = \herm{\mF}_{k, a, b}^2$ and $\mA_{i, l} = \frac{1}{\sqrt{\dpoly}} \mPi_{i, l}\mSig_i$ where we observe that
	\[\sum_{i = 1}^n\sum_{l = 1}^d \mA_{i, l} \mA_{i, l}^T = \frac{1}{\dpoly}\sum_{i = 1}^n\sum_{l = 1}^d \mPi_{i, l}\mSig_i\mSig_i^\T\mPi_{i, l}^\T = \frac{1}{\dpoly}\sum_{i = 1}^n\sum_{l = 1}^d \mPi_{i, l}^2 \preceq \mI\]
	to get
	\begin{align*}
		\Esch{\mDel_{11}}{t} = \EE \tr[\mDel_{11}^t] &\le \frac{d^t}{n^t} \EE\tr[(\sum_{i = 1}^n\sum_{l = 1}^d \mPi_{i, l}\mSig_i\herm{\mF}_{k, a, b}^2\mSig_i^\T\mPi_{i, l})^t]\\
		&\le \frac{(d\dpoly)^t}{n^t}\EE\tr[\bigg(\frac{1}{\dpoly}\sum_{i = 1}^n\sum_{l = 1}^d \mPi_{i, l}\mSig_i\herm{\mF}_{k, a, b}^{2t}\mSig_i^\T\mPi_{i, l}\bigg)]\\
		&= \frac{(d\dpoly)^t}{n^t}\EE\tr[\bigg(\frac{1}{\dpoly}\sum_{i = 1}^n\sum_{l = 1}^d \mSig_i^\T\mPi_{i, l}\mPi_{i, l}\mSig_i\herm{\mF}_{k, a, b}^{2t}\bigg)]
	\end{align*}
	To simplify this, we use \cref{fact: commutativity} to get
	\[\sum_{i = 1}^n\sum_{l = 1}^d \mSig_i^\T(\mPi_{i, l})^2\mSig_i \preceq \sum_{i = 1}^n\sum_{l = 1}^d \mSig_i^\T(\mPi'_{i, l})^2\mSig_i = \sum_{i = 1}^n\sum_{l = 1}^d \mPi'_{i, l}\mSig_i^\T\mSig_i\mPi'_{i, l} = \sum_{i = 1}^n\sum_{l = 1}^d \mPi'_{i, l}\mPi'_{i, l} \preceq \dpoly \mI\]
	Therefore,
	$\Esch{\mDel_{11}}{t} \le  \frac{(d\dpoly)^t}{n^t}\EE\tr[\herm{\mF}_{k, a, b}^{2t}] = \frac{(d\dpoly)^t}{n^t}\Esch{\herm{\mF}_{k, a, b}}{2t}$.
	Putting them together and using \cref{fact: holder},
	\begin{align*}
		\Esch{\mDel_1^{k, a, b}}{t} \le 4^t(\Esch{\mDel_{10}}{t} + \Esch{\mDel_{11}}{t})
		\le \frac{(8d\dpoly)^t}{n^t}\Esch{\herm{\mF}_{k, a, b}}{2t}
	\end{align*}
\end{proof}

We now restate and prove \cref{lem: bound_on_Del3}.

\boundDelThree*

\begin{proof}
	Recall that $Z' = Z^{(i)}$ for $i$ sampled uniformly from $[n]$. Then,
	\begin{align*}
		\mDel_3^{k, a, b} &= \EE[(\mD(Z)\herm{\mG}_{k, a, b}(Z)(\mD(Z) - \mD(Z')))^2|Z]\\
		&= \EE[\EE_{i \in [n]} [(\mD(Z)\herm{\mG}_{k, a, b}(Z)(\mD(Z) - \mD(Z^{(i)})))^2] | Z]\\
		&= \EE[\EE_{i \in [n]} [(\mD(Z)\herm{\mG}_{k, a, b}(Z)\mPi_i(\mD(Z) - \mD(Z^{(i)})))^2] | Z]
	\end{align*}
	where we use the fact that $\mD(Z) - \mD(Z^{(i)}) = \mPi_i(\mD(Z) - \mD(Z^{(i)}))$ for all $i$. Define $\mM(Z) = \mD(Z)\herm{\mG}_{k, a, b}$ to get
	\begin{align*}
		\mDel_3^{k, a, b} &= \EE[\EE_{i \in [n]} [\mM(Z) \mPi_i(\mD(Z) - \mD(Z^{(i)}))^2\mPi_i\mM(Z)^\T] | Z]\\
		&\preceq 2(\EE[\EE_{i \in [n]} [\mM(Z) \mPi_i\mD(Z)^2\mPi_i\mM(Z)^\T] | Z] + \EE[\EE_{i \in [n]} [\mM(Z) \mPi_i\mD(Z^{(i)})^2\mPi_i\mM(Z)^\T] | Z])\\
		&= 2(\EE_{i \in [n]} [\mM(Z) \mPi_i\mD(Z)^2\mPi_i\mM(Z)^\T] + \EE[\EE_{i \in [n]} [\mM(Z) \mPi_i\mD(Z^{(i)})^2\mPi_i\mM(Z)^\T] | Z])\\
		&= 2(\mDel_{30} + \mDel_{31})
	\end{align*}
	where we define
	\[\mDel_{30} = \EE_{i \in [n]} [\mM(Z) \mPi_i\mD(Z)^2\mPi_i\mM(Z)^\T], \qquad \mDel_{31} = \EE[\EE_{i \in [n]} [\mM(Z) \mPi_i\mD(Z^{(i)})^2\mPi_i\mM(Z)^\T] | Z]\]
	We have
	\begin{align*}
		\mDel_{30} = \EE_{i \in [n]} [\mM(Z) \mPi_i\mD(Z)^2\mPi_i\mM(Z)^\T] &= \EE_{i \in [n]} [\mM(Z) \mD(Z)\mPi_i\mPi_i\mD(Z)\mM(Z)^\T]\\
		&= \mM(Z) \mD(Z)(\frac{1}{n}\sum_{i = 1}^n\mPi_i^2)\mD(Z)\mM(Z)^\T\\
		&\preceq \frac{\dpoly}{n} \mM(Z) \mD(Z)\mD(Z)\mM(Z)^\T\\
		&= \frac{\dpoly}{n} \herm{\mF}_{k, a, b}^2
	\end{align*}
	For the other term, using \cref{fact: commutativity},
	\begin{align*}
		\mDel_{31} = \EE[\EE_{i \in [n]} [\mM(Z) \mPi_i\mD(Z^{(i)})^2\mPi_i\mM(Z)^\T] | Z]
		&= \EE_{i \in [n]} [\mM(Z) \mPi_i\mSig_i\mD(Z)^2\mSig_i\mPi_i\mM(Z)^\T]\\
		&\preceq \EE_{i \in [n]} [\mM(Z) \mPi'_i\mSig_i\mD(Z)^2\mSig_i\mPi'_i\mM(Z)^\T]\\
		&= \EE_{i \in [n]} [\mM(Z) \mSig_i\mPi'_i\mD(Z)^2\mPi'_i\mSig_i\mM(Z)^\T]\\
        &= \EE_{i \in [n]} [\mD(Z)\herm{\mG}_{k, a, b} \mSig_i\mPi'_i\mD(Z)^2\mPi'_i\mSig_i\herm{\mG}_{k, a, b}\mD(Z)]
	\end{align*}
	Observe that $\herm{\mG}_{k, a, b} \mSig_i = \herm{\mG}_{k, a, b}$ because the entries of $\herm{\mG}$ only depend on $\al$ and not on $\gam$, so permuting the $\gamma$s will not have any effect on the matrix. Therefore,
	\begin{align*}
		\mDel_{31} &\preceq \EE_{i \in [n]} [\mD(Z)\herm{\mG}_{k, a, b}\mPi'_i\mD(Z)^2\mPi'_i\herm{\mG}_{k, a, b}\mD(Z)]\\
		&\preceq \EE_{i \in [n]} [\mD(Z)\herm{\mG}_{k, a, b}\mD(Z)\mPi'_i\mPi'_i\mD(Z)\herm{\mG}_{k, a, b}\mD(Z)]\\
		&= \EE_{i \in [n]} \herm{\mF}_{k, a, b}\mPi'_i\mPi'_i\herm{\mF}_{k, a, b}\\
		&= \frac{1}{n} \sum_{i = 1}^n \herm{\mF}_{k, a, b}\mPi'_i\mPi'_i\herm{\mF}_{k, a, b}\\
		&\preceq \frac{\dpoly}{n}\herm{\mF}_{k, a, b}^2
	\end{align*}
	where we used the fact that $\sum_{i = 1}^n\mPi'_i\mPi'_i \preceq \dpoly \mI$. Putting them together,
	\begin{align*}
		\Esch{\mDel_3^{k, a, b}}{t} \le 2^t(\Esch{\mDel_{30}}{t} + \Esch{\mDel_{31}}{t}) \le 2^t \cdot 2 \frac{\dpoly^t}{n^t}\Esch{\herm{\mF}_{k, a, b}}{2t} \le \frac{(4\dpoly)^t}{n^t}\Esch{\herm{\mF}_{k, a, b}}{2t}
	\end{align*}
\end{proof}



%% file: sparse_graph_matrices.tex
We now consider sparse graph matrices, i.e., the setting $G \sim \cG_{n, p}$ for $p \le \frac{1}{2}$.
The main difference from dense graph matrices is the contribution of the edge factors.  Na\"ively bounding the contribution of each edge by it's absolute value, as explained in \cref{sec: failure_of_basic}, each edge in the shape contributes a factor of $\sqrt{\frac{1 - p}{p}}$. But in many cases, these bounds are not tight. In fact, they are not tight even in the basic case of the adjacency matrix. In this section, we obtain tighter bounds using our general recursion. As we will see, the improved bound will contain the edge factors only for edges within the vertex separator.

Let $\graphmat{\tau}$ be the graph matrix corresponding to shape $\tau$ where we use $p$-biased Fourier characters $G_{i, j}$. In this section, we obtain bounds on $\Esch{\mM_{\tau} - \EE \mM_{\tau}}{2t}$ and use it to obtain high probability bounds on $\norm{\mM_{\tau}}$.
Since many of the details are similar to \cref{sec: norm_bounds_for_dense_graph_matrices} and the proof of \cref{thm: dense_graph_matrix_norm_bounds}, we will pass lightly over some details. We recommend the reader to read that section first.

The $G_{i, j}$ correspond to the $Z_i$s in \cref{sec: general_recursion} and $\mF$ corresponds to $\mM_{\tau}$. Let $\cI$ denote the set of sub-tuples of $[n]$. Each nonzero entry of $\mM_{\tau}$ is a homogenous polynomial of degree $|E(\tau)|$. If $E(\tau) = \emptyset$, then, $\mM_{\tau} - \EE \mM_{\tau} = 0$ so we can focus on the case when $\tau$ has at least one edge. Moreover, since degree-$0$ vertices in $V(\tau)\setminus U_{\tau} \setminus V_{\tau}$ simply scale the matrix by a factor of at most $n$, we can handle them separately and for our main analysis, we assume there are no such vertices in $\tau$.

We will use \cref{thm: main_general} but the matrices and the statement can be drastically simplified in our application. Instate the notation of \cref{sec: general_recursion}. Since we are dealing with multilinear polynomials, in the definition of $\cK$, we can restrict our attention to $\al \in \{0, 1\}^{\binom{n}{2}}$ because for any other $\al \in \NN^n$, the corresponding row or column of $\mG_{a + b, a, b}$ and hence $\mF_{a + b, a, b}$, will be $0$. So, we can accordingly redefine $\cK$ to only contain these $(\al, \gam)$, hence $\cK \subseteq \{0, 1\}^n \times \{0, 1\}^n$.

Next, the diagonal matrices $\mD_1, \mD_2$ will both be equal to the diagonal matrix $\mD \in \RR[Z]^{\cI \times \cK} \times \RR[Z]^{\cI \times \cK}$ with nonzero entries
\[\mD[(I, \al, \gam), (I, \al, \gam)] = \sqrt{\EE[ \prod_{i, j} G_{ij}^{2\al_{ij}(1 - \gam)_{ij}}]} \prod_{i, j} G_{i}^{\al_{ij}\gam_{ij}} =  \prod_{i, j} G_{i}^{\al_{ij}\gam_{ij}}\]
where we used the fact that for any $i, j$, $\EE[G_{ij}^2] = 1$.

For integers $a, b \ge 0$ such that $a + b = |E(\tau)|$, define the matrix $\mM_{\tau, a, b}$ to be the matrix $\mG_{a + b, a, b}$. We use this notation in order to be streamlined with \cref{sec: norm_bounds_for_dense_graph_matrices}. That is, $\mM_{\tau, a, b}$ has rows and columns indexed by $\cI\times \cK$ such that for all $(I, \al_1, \gam_1), (J, \al_2, \gam_2) \in \cI \times \cK$,
\[\mM_{\tau, a, b}[(I, \al_1, \gam_1), (J, \al_2, \gam_2)] = \begin{dcases}
	\grad_{\al_1 + \al_2} \mM_{\tau}[I, J] & \text{ if $|\al_1|_0 = a, |\al_2|_0 = b, \al_1 \cdot \al_2 = 0$}\\
	0 & \text{o.w.}
\end{dcases}
\]

This is almost identical to the $\mM_{\tau, a, b}$ matrix defined in \cref{sec: norm_bounds_for_dense_graph_matrices}, with the difference being that the row and column indices now have $\gam$ in them. Therefore, for $I, J \in \cI, (\al_1, \gam_1), (\al_2, \gam_2) \in \cK$ such that $|\al_1|_0 = a, |\al_2|_0 = b, \al_1 \cdot \al_2 = 0$, the entry in row $(I, \al_1, \gam_1)$ and column $(J, \al_2, \gam_2)$ is the number of realizations $\phi$ of $\tau$ such that
\begin{itemize}
	\item $U_{\tau}, V_{\tau}$ map to $I, J$ respectively under $\phi$, and
	\item Under $\phi$, the edges of $\tau$ map to the edges in $\al_1$ and $\al_2$ viewed as a set.
\end{itemize}

By \cref{thm: main_general}, for integers $t \ge 1$,
\begin{align*}
	\Esch{\mM_{\tau} - \EE \mM_{\tau}}{2t} &\le \sum_{a, b \ge 0, a + b \ge 1}(Ct^2d\dpoly^4)^{(a + b)t}\Esch{\mF_{a + b, a, b}}{2t}\\
	&= \sum_{a, b \ge 0, a + b  = |E(\tau)|}(Ct^2|E(\tau)|^4)^{t|E(\tau)|}\Esch{\mD\mM_{\tau, a, b}\mD}{2t}
\end{align*}
for an absolute constant $C > 0$.

Now, we would like to analyze $\Esch{\mD\mM_{\tau, a, b}\mD}{2t}$. Just as in the proof of \cref{thm: dense_graph_matrix_norm_bounds}, let $P$ specify which edges of $E(\tau)$ go to $\al_1, \al_2$ respectively and in what order. Moreover, we now store extra information in $P$ that indicates which entries of $\gam_1, \gam_2$ (relative to $\al_1, \al_2$) are set to $1$. Let the set of such information $P$ be denoted $\cP$, then $|\cP| \le (4|E(\tau)|)^{t|E(\tau)|} 2^{|E(\tau)|}$. Thus,
\begin{align*}
	\Esch{\mD\mM_{\tau, a, b}\mD}{2t} \le (8|E(\tau)|)^{t|E(\tau)|} \sum_{P \in \cP} \Esch{\mD\mM_{\tau, a, b, P}\mD}{2t}
\end{align*}
where we define $\mM_{\tau, a, b, P}$ similar to $\mM_{\tau, a, b}$ with the extra condition that $\phi, \al_1, \al_2, \gam_1, \gam_2$ must respect $P$.

At this point, in contrast to the proof of \cref{thm: dense_graph_matrix_norm_bounds}, note that the matrices $\mM_{\tau, a, b, P}$ here have rows and columns indexed by $\cI\times \cK$. We will again define the shape $\tau_P$ that is equal to the nonzero block of the matrix $\mD\mM_{\tau, a, b, P}\mD$, up to renaming of the rows and columns. $V(\tau_P), U_{\tau_P}, V_{\tau_P}$ are defined the same way as in \cref{sec: norm_bounds_for_dense_graph_matrices} but to incorporate the action of $\mD$ on these entries, we simply keep the edges that are active in $\gam_1$ or $\gam_2$, as prescribed by $P$. For an illustration, see \cref{fig: evolution_sparse}.

\begin{figure}[!h]
	\centering
	\includegraphics[trim={2cm 20cm 0 2cm}, clip, scale=0.9]{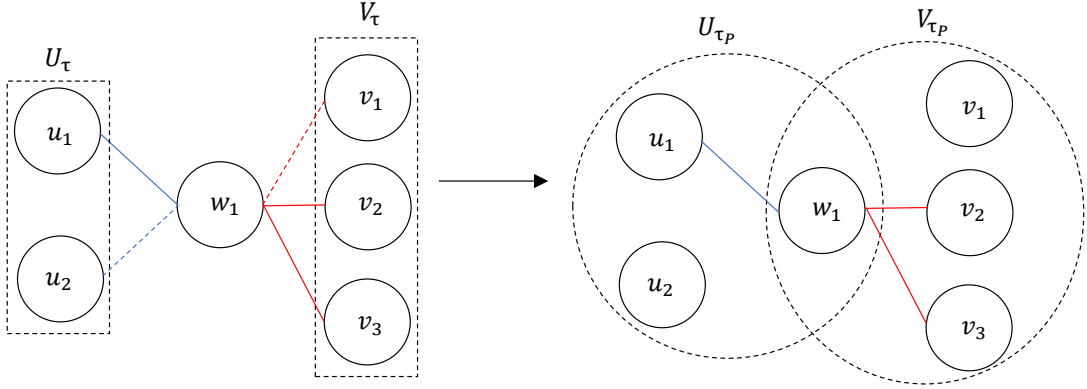}
	\caption{An example illustrating how $\tau_P$ is defined. In this example, $P$ constraints the blue and red edges to go to $\al_1$ and $\al_2$ respectively. Moreover, $P$ indicates that some edges are active in $\gam_1, \gam_2$ (indicated by a solid edge) and some are not active (indicated by a dashed edge) in $\gam_1, \gam_2$. We keep the solid edges in $\tau_P$. $U_{\tau_P}, V_{\tau_P}$ also have an ordering on the vertices (not shown here).}
	\label{fig: evolution_sparse}
\end{figure}

Then, by similar renaming of the rows and columns of $\mD\mM_{\tau, a, b, P}\mD$ and dropping the $\gam$s, we obtain $\mM_{\tau_P}$. We therefore obtain the bound
\begin{align*}
	\Esch{\mD\mM_{\tau, a, b}\mD}{2t} &\le (8|E(\tau)|)^{t|E(\tau)|} \sum_{P \in \cP}\Esch{\mM_{\tau_P}}{2t}
\end{align*}

We would like to analyze norm bounds on the matrices $\mM_{\tau_P}$. Observe that $\tau_P$ are shapes with the properties
\begin{itemize}
	\item there are no vertices in $V(\tau_P) \setminus U_{\tau_P} \setminus V_{\tau_P}$
	\item each edge is either entirely contained in $U_{\tau_P}$ or entirely contained in $V_{\tau_P}$
\end{itemize}
Call such shapes \textit{simple}.

In the following lemma, whose proof is deferred to the next section, we prove norm bounds on simple shapes. Recall that in \cref{lem: empty_shape}, we analyzed the norm bounds of simple shapes with no edges (because in this case, the graph distribution doesn't matter). The analysis for simple shapes is very similar but this time, we use scalar concentration tools to bound the Frobenius norm.

For a set $S$ of vertices, denote by $E(S)$ the set of edges with both endpoints in $S$.

\begin{restatable}{lemma}{simplegraphmatrixnormbounds}\label{lem: graphmatrixnormbound_nomiddlevertices}
	For all even integers $t \ge 2$, if $\tau$ is a simple shape,
	\[\Esch{\graphmat{\tau}}{2t} \le \bigg(n^{|V(\tau)|} (Ct)^{t|E(\tau)|} |V(\tau)|^{t|V(\tau)|}\bigg)\max_{U_{\tau} \cap V_{\tau} \subseteq S \subseteq V(\tau)}\left(\frac{1 - p}{p}\right)^{t|E(S)|}n^{t(|V(\tau)| - |S|)}\]
	for an absolute constant $C > 0$.
\end{restatable}

For simple shapes, the main difference from norm bounds on corresponding dense graph matrices is that each edge within $S$ contributes a factor of $\sqrt{\frac{1 - p}{p}}$. Edge contributions are unavoidable when handling sparse graph matrices, but we have identified that we need not consider all edges in the shape but only a subset of it.
Using this lemma, we can obtain norm bounds on general graph matrices. We first recall the definition of a vertex separator.

\vertexseparator*

Let $I_{\tau}$ be the set of isolated vertices (vertices of degree $0$) in $V(\tau) \setminus U_{\tau} \setminus V_{\tau}$, so they essentially scale the matrix by a scalar factor. We now state the main theorem of this section.

\begin{theorem}\label{thm: sparse_graph_matrix_norm_bounds}
	For all even integers $t \ge 2$, for any shape $\tau$,
	\[\Esch{\mM_{\tau} - \EE\mM_{\tau}}{2t} \le \bigg(n^{|V(\tau)|} |V(\tau)|^{t|V(\tau)|} (Ct^3|E(\tau)|^5)^{t|E(\tau)|}\bigg) \max_{\text{vertex separator }S}\left(\frac{1 - p}{p}\right)^{t|E(S)|}n^{t(|V(\tau)| - |S| + |I_{\tau}|)}\]
	where the maximum is over all vertex separators $S$.
\end{theorem}

To interpret this bound, if we assume that there are a constant number of vertices in $\tau$, then by choosing $t \approx \polylog(n)$, we get
\[\norm{\mM_{\tau}} = \widetilde{\bigoh}\bigg(\max_{\text{vertex separator }S}\left(\sqrt{\frac{1 - p}{p}}\right)^{|E(S)|}\sqrt{n}^{|V(\tau) - |S| + |I_{\tau}|}\bigg)\] with high probability, where $\widetilde{\bigoh}$ hides logarithmic factors.  This result follows from \cref{thm: sparse_graph_matrix_norm_bounds} if $\tau$ has at least one edge, but also applies if $\tau$ has no edges, in which case we can directly use the far simpler \cref{lem: empty_shape}. A precise form of the above characterization is given in \cref{cor: sparse_graph_matrix_norm_bounds}.

\cref{thm: sparse_graph_matrix_norm_bounds} gives us the right dependence on $p, n$ for norm bounds in the case of sparse graph matrices. The same bound, up to lower order terms, was also obtained in \cite{sparsesos} via the trace power method, where they use these bounds to prove semidefinite-programming lower bounds for the maximum independent set problem on sparse graphs.

\begin{proof}[Proof of \cref{thm: sparse_graph_matrix_norm_bounds}]
	If $E(\tau) = \emptyset$, then $\mM_{\tau} = \EE\mM_{\tau}$ and we are done. So, assume $E(\tau) \neq \emptyset$. Since vertices in $I_{\tau}$ only scale the matrix by a factor of at most $n$, we can handle them separately and our bound has the appropriate power of $n$ coming from these. Therefore, we can assume $I_{\tau} = \emptyset$. Continuing our prior discussions, for an absolute constant $C_1 > 0$,
	\begin{align*}
		\Esch{\mM_{\tau} - \EE\mM_{\tau}}{2t} &\le \sum_{a, b \ge 0, a + b  = |E(\tau)|}(C_1t^2|E(\tau)|^4)^{t|E(\tau)|}\Esch{\mD\mM_{\tau, a, b}\mD}{2t}\\
		&\le \sum_{a, b \ge 0, a + b  = |E(\tau)|}(C_1t^2|E(\tau)|^4)^{t|E(\tau)|}(8|E(\tau)|)^{t|E(\tau)|} \sum_{\psi \in \Gam_{a, b}}\Esch{\mM_{\psi}}{2t}
	\end{align*}
	where $\Gam_{a, b}$ are the set of simple shapes we obtain for $\mD\mM_{\tau, a, b} \mD$, as per our discussion above. Using \cref{lem: graphmatrixnormbound_nomiddlevertices}, for an absolute constant $C_2 > 0$, we have
	\begin{align*}
		&\Esch{\mM_{\tau} - \EE\mM_{\tau}}{2t}\\
		&\le \bigg(n^{|V(\tau)|} |V(\tau)|^{t|V(\tau)|} (C_2t^3|E(\tau)|^5)^{t|E(\tau)|}\bigg)\sum_{a, b \ge 0, a + b  = |E(\tau)|}\sum_{\psi \in \Gam_{a, b}}\max_{U_{\psi} \cap V_{\psi} \subseteq S \subseteq V(\psi)}\left(\frac{1 - p}{p}\right)^{t|E(S)|}n^{t(|V(\psi)| - |S|)}
	\end{align*}
	For any $a, b$, consider any simple shape $\psi \in \Gam_{a, b}$ that can be obtained. As observed in the proof of \cref{thm: dense_graph_matrix_norm_bounds} (see in particular \cref{fig: proof_by_picture}), $U_{\psi} \cap V_{\psi}$ must be a vertex separator of $\tau$. Therefore, any $S \supseteq U_{\psi} \cap V_{\psi}$ must be a vertex separator of $\tau$. It's easy to see that as $S$ ranges over all sets such that $U_{\psi} \cap V_{\psi} \subseteq S \subseteq V(\psi)$, it ranges over all vertex separators of $\tau$.

	Also, the number of different $\psi$ is at most $4^{|E(\tau)|}$ since each edge can go either to $U_{\psi}$ or $V_{\psi}$ and for each such choice, it can either be active in $\gam$ or not. Therefore,
	\begin{align*}
		&\Esch{\mM_{\tau} - \EE\mM_{\tau}}{2t}\\
		&\le \bigg(n^{|V(\tau)|} |V(\tau)|^{t|V(\tau)|} (C_2t^3|E(\tau)|^5)^{t|E(\tau)|}\bigg) 4^{|E(\tau)|}\max_{\text{vertex separator }S}\left(\frac{1 - p}{p}\right)^{t|E(S)|}n^{t(|V(\tau)| - |S|)}\\
		&\le \bigg(n^{|V(\tau)|} |V(\tau)|^{t|V(\tau)|} (Ct^3|E(\tau)|^5)^{t|E(\tau)|}\bigg) \max_{\text{vertex separator }S}\left(\frac{1 - p}{p}\right)^{t|E(S)|}n^{t(|V(\tau)| - |S|)}
	\end{align*}
	for an absolute constant $C > 0$.
\end{proof}

The following corollary obtains high probability norm bounds for norms of graph matrices via Markov's inequality. We assume the graph has at least one edge, otherwise it is deterministic and its norm bound was already analyzed in \cref{lem: empty_shape}, \cref{cor: dense_graph_matrix_norm_bounds}, where we observe that the distinction between sparse and dense graph matrices does not matter if the random matrix is deterministic.

\begin{corollary}\label{cor: sparse_graph_matrix_norm_bounds}
	For a shape $\tau$ with at least one edge, for any constant $\eps > 0$, with probability $1 - \eps$,
	\[\norm{\mM_{\tau}} \le \bigg(|V(\tau)|^{|V(\tau)|/2} (C|E(\tau)|^5\log^3(n^{|V(\tau)|}/\eps))^{|E(\tau)|/2}\bigg)\cdot\max_{\text{vertex separator }S}\left(\sqrt{\frac{1 - p}{p}}\right)^{|E(S)|}\sqrt{n}^{|V(\tau) - |S| + |I_{\tau}|}\]
	for an absolute constant $C > 0$.
\end{corollary}

\begin{proof}
	Since $|E(\tau)| \ge 1$, $\EE\mM_{\tau} = 0$. By an application of Markov's inequality,
	\begin{align*}
		Pr[\norm{\mM_{\tau}} \ge \theta] &\le Pr[\sch{\mM_{\tau}}{2t} \ge \theta^{2t}]\\
		&\le \theta^{-2t} \EE\sch{\mM_{\tau}}{2t}\\
		&\le \theta^{-2t}\bigg(n^{|V(\tau)|} |V(\tau)|^{t|V(\tau)|} (C't^3|E(\tau)|^5)^{t|E(\tau)|}\bigg) \max_{\text{vertex separator }S}\left(\frac{1 - p}{p}\right)^{t|E(S)|}n^{t(|V(\tau)| - |S| + |I_{\tau}|)}
	\end{align*}
	for an absolute constant $C' > 0$. We now set
	\begin{align*}
		\theta = &\bigg(\eps^{-1/(2t)} (C'')^{|E(\tau)|} n^{|V(\tau)|/(2t)}|V(\tau)|^{|V(\tau)|/2} t^{3|E(\tau)|/2}|E(\tau)|^{5|E(\tau)|/2}\bigg)\\
		&\qquad\cdot \max_{\text{vertex separator }S}\left(\sqrt{\frac{1 - p}{p}}\right)^{|E(S)|}\sqrt{n}^{|V(\tau) - |S| + |I_{\tau}|}
	\end{align*}
	for an absolute constant $C'' > 0$, to make this expression at most $\eps$. Set $t = \frac{1}{2}\log(n^{|V(\tau)|}/\eps)$ to complete the proof.
\end{proof}

\subsection{Norm bounds on simple graph matrices}

In this section, we will prove \cref{lem: graphmatrixnormbound_nomiddlevertices}. First, we recall the following scalar concentration result from \cite{schudy2011bernstein}.

\subsubsection{Schudy-Sviridenko moment bound}

The definitions and main bound in this section are from \cite{schudy2011bernstein}.

\begin{definition}
	A random variable $Z$ is central moment bounded with real parameter $L > 0$ if for any integer $i \ge 1$,
	\[\EE[|Z - \EE[Z]|^i] \le i\cdot L\cdot \EE[|Z - \EE[Z]|^{i - 1}]\]
\end{definition}

\begin{propn}
	The $p$-biased Bernoulli random variable $Z$ is central moment bounded with real parameter $L =\sqrt{\frac{1 - p}{p}}$.
\end{propn}

\begin{proof}
	We have $\EE[Z] = 0$ and for $p \le \frac{1}{2}$, $|Z| \le \sqrt{\frac{1 - p}{p}}$, therefore,
	\begin{align*}
		\EE[|Z - \EE[Z]|^i] &= p\sqrt{\frac{p}{1 - p}}^i + (1 - p)\sqrt{\frac{1 - p}{p}}^i\\
		&\le \sqrt{\frac{1 - p}{p}}\bigg(p\sqrt{\frac{p}{1 - p}}^{i - 1} + (1 - p)\sqrt{\frac{1 - p}{p}}^{i- 1}\bigg)\\
		&= \sqrt{\frac{1 - p}{p}}\EE[|Z - \EE[Z]|^{i - 1}]
	\end{align*}
	therefore, we can take $L = \sqrt{\frac{1 - p}{p}}$.
\end{proof}

For a given multilinear polynomial $f(x)$ on variables $x_1, \ldots, x_n$, we can naturally associate with it a hypergraph $H$ on vertices $[n]$ and weighted hyperedges $E(H)$ where each $h \in E(H)$ corresponds to a distinct term of $f(x)$. Each hyperedge $h$ is a subset $V(h)$ of vertices and has a real valued weight $w_h$ which is the coefficient of that monomial in $f$. Therefore,
\[f(x) = \sum_{h \in E(H)} w_h \prod_{v \in V(h)} x_v\]

Assume $f$ has degree $\dpoly$, then each hyperedge of $H$ has at most $\dpoly$ vertices.
Now, for a given collection of independent random variables $Y_1, \ldots, Y_n$, a multilinear poynomial $f$ with associated hypergraph $H$ and weights $w$, and an integer $r \ge 0$, define
\[\mu_r(f, Y) = \max_{S \subseteq [n], |S| = r} \bigg(\sum_{h \in E(H), S \subseteq V(h)} |w_h|\prod_{v \in V(h) \setminus S} \EE[|Y_v|]\bigg)\]

\begin{lemma}[\cite{schudy2011bernstein}, Lemma 5.1]\label{lem: schudy_sviridenko}
	Given $n$ independent central moment bounded random variables $Y_1, \ldots, Y_n$ with the same parameter $L > 0$ and a degree $\dpoly$ multilinear polynomial $f(x)$. Let $t \ge 2$ be an even integer, then
	\[\EE[|f(Y) - \EE [f(Y)]|^t] \le \max\bigg\{\bigg(\sqrt{tR_4^{\dpoly}\var{f(Y)}}\bigg)^t, \max_{r \in [\dpoly]}(t^rR_4^{\dpoly}L^r\mu_r(f, Y))^t\bigg\}\]
	where $R_4 \ge 1$ is some absolute constant.
\end{lemma}

In our setting, we can also bound the variance in terms of the $\mu_r$ as was shown in \cite{schudy2011bernstein}, which will simplify our calculations.

\begin{lemma}[\cite{schudy2011bernstein}, Lemma 1.5]\label{lem: var_bound}
	For the same setting as in \cref{lem: schudy_sviridenko},
	\[\var{f(Y)} \le 2\dpoly 4^{\dpoly}\max_{r \in [\dpoly]} (\mu_0(f, Y) \mu_r(f, Y)4^rL^r)\]
\end{lemma}

\subsubsection{Proof of \cref{lem: graphmatrixnormbound_nomiddlevertices}}

We are ready to prove \cref{lem: graphmatrixnormbound_nomiddlevertices} which we restate for convenience.

\simplegraphmatrixnormbounds*

We will prove it the same way as \cref{lem: empty_shape}, by bounding the schatten norm of each diagonal block by an appropriate power of its Frobenius norm. In this case, to bound the expected power of the Frobenius norm, we use the scalar concentration inequality from the previous section.

\begin{proof}[Proof of \cref{lem: graphmatrixnormbound_nomiddlevertices}]
	First, we note that $\mM_{\tau}$ has a block diagonal structure indexed by the realizations of the set of common vertices $S_0 = U_{\tau_P} \cap V_{\tau_P}$. For $T \in [n]^{S_0}$, let $\mM_{\tau, T}$ be the block of $\mM_{\tau}$ with $\phi(S_0) = T$. Then, $\mM_{\tau, T}\mM_{\tau, T'}^\T = \mM_{\tau, T}^\T\mM_{\tau, T'} = 0$ for $T \neq T'$ and so,
	\begin{align*}
		\Esch{\mM_{\tau}}{2t} = \sum_{T \in [n]^{S_0}} \Esch{\mM_{\tau, T}}{2t} \le \sum_{T \in [n]^{S_0}}\EE(\sch{\mM_{\tau, T}}{2})^t
	\end{align*}
	where we bounded the Schatten norm by a power of the Frobenius norm.

	Fix $T \in [n]^{S_0}$ and consider $\Esch{\mM_{\tau, T}}{2}$. Let $\cR$ be the set of realizations $\phi$ of $\tau$ such that $\phi(S_0) = T$. Then, for $\phi \in \cR$ and $e \in E(S_0)$, the value of $\phi(e)$ is fixed. Using this,
	\begin{align*}
		\norm{\mM_{\tau, T}}_2^2 = \sum_{\phi \in \cR} \prod_{e \in E(\tau)} G_{\phi(e)}^2
		&= \prod_{e \in E(S_0)} G_{\phi(e)}^2 \sum_{\phi \in \cR} \prod_{e \in E(\tau) \setminus E(S_0)} G_{\phi(e)}^2\\
		&\le L^{|E(S_0)|} \sum_{\phi \in \cR} \prod_{e \in E(\tau) \setminus E(S_0)} G_{\phi(e)}^2
	\end{align*}
	where $L = \frac{1 - p}{p}$ is an upper bound on $G_{ij}^2$ for $p \le \frac{1}{2}$. For convenience, we define the quantity
	$A = \max_{S_0 \subseteq S \subseteq V(\tau)}L^{|E(S)|}n^{|V(\tau)| - |S|}$.

	\begin{claim}\label{claim: lil_claim}
		$\EE(\norm{\mM_{\tau, T}}_2)^t \le (Ct)^{t|E(\tau)|} |V(\tau)|^{t|V(\tau)|}A^t$ for an absolute constant $C > 0$.
	\end{claim}

	Using this claim, we have
	\begin{align*}
		\Esch{\mM_{\tau}}{2t} &\le \sum_{T \in [n]^{S_0}} \EE(\norm{\mM_{\tau, T}}_2)^t\\
		&\le n^{|S_0|}(Ct)^{t|E(\tau)|} |V(\tau)|^{t|V(\tau)|}A^t\\
		&= n^{|V(\tau)|} (Ct)^{t|E(\tau)|} |V(\tau)|^{t|V(\tau)|}\max_{U_{\tau} \cap V_{\tau} \subseteq S \subseteq V(\tau)}\left(\frac{1 - p}{p}\right)^{t|E(S)|}n^{t(|V(\tau)| - |S|)}
	\end{align*}
	as required.
\end{proof}

It remains to prove the claim.

\begin{proof}[Proof of \cref{claim: lil_claim}]
	For $1\le i, j \le n$, define the variables $Y_{ij} = G_{ij}^2$ with $\EE[|Y_{ij}|] = 1$. Let $f(Y)$ be the polynomial $L^{|E(S_0)|} \sum_{\phi \in \cR} \prod_{e \in E(\tau) \setminus E(S_0)} Y_{\phi(e)}$. It suffices to prove that $\EE[f(Y)^t] \le (Ct)^{t|E_1|}A^t$.
	We will first prove that $\EE[(f(Y) - \EE[f(Y)])^t] \le (C't)^{t|E(\tau)|}|V(\tau)|^{t|V(\tau)|}A^t$ for a sufficiently large constant $C' > 0$.

	$f$ is a homogeneous multilinear polynomial of degree $|E(\tau) \setminus E(S_0)|$. If we had $E(\tau) \setminus E(S_0) = \emptyset$, then $f$ is a constant and so, the inequality is obvious because $f(Y) = \EE[f(Y)]$. Now, assume $E(\tau) \setminus E(S_0) \neq \emptyset$. We invoke \cref{lem: schudy_sviridenko}. Let $f$ have associated hypergraph $H$ and weights $w$. Then,
	\[\EE[|f(Y) - \EE [f(Y)]|^t] \le \max\bigg\{\bigg(\sqrt{tR_4^{|E(\tau) \setminus E(S_0)|}\var{f(Y)}}\bigg)^t, \max_{r \in [|E(\tau) \setminus E(S_0)|]}(t^rR_4^{|E(\tau) \setminus E(S_0)|}L^r\mu_r(f, Y))^t\bigg\}\]
	For all $r \ge 0$, we will prove that $L^r\mu_r(f, Y) \le |V(\tau)|^{|V(\tau)|}A$. By definition,
	\begin{align*}
		\mu_r(f, Y) &= \max_{F \subseteq \binom{[n]}{2}, |F| = r} \sum_{h \in E(H), F \subseteq V(h)} |w_h|
	\end{align*}
	Consider any set of edge labels $F \subseteq \binom{[n]}{2}, |F| = r$.
	Then, $\sum_{h \in E(H), F \subseteq V(h)} |w_h|$ is at most $L^{|E(S_0)|}c$ where $c$ is the number of realizations $\phi \in \cR$ such that $\phi(E(\tau))$ contains $F$.
	Suppose $F$ contains $v$ new labels apart from $\phi(S_0) = T$.
	Then $c \le |V(\tau)|^v n^{|V(\tau)| - |S_0| - v}$ because we can first choose and label the set of vertices that get these $v$ labels and then label the remaining vertices freely, each of which has at most $n$ choices.

	Observe that $L^{|E(S_0)|} L^r n^{|V(\tau)| - |S_0| - v} \le A$ because in the definition of $S$, we can set $S$ to be the union of $S$ and any valid choice of these $v$ vertices. Putting this together, we get
	\begin{align*}
		L^r\mu_r(f, Y) \le L^r\max_{F \subseteq \binom{[n]}{2}, |F| = r} \sum_{h \in E(H), F \subseteq V(h)} |w_h|
		\le |V(\tau)|^{|V(\tau)|} A
	\end{align*}
	which implies
	\[\max_{r \in [|E(\tau) \setminus E(S_0)|]}(t^rR_4^{|E(\tau) \setminus E(S_0)|}L^r\mu_r(f, Y))^t \le |V(\tau)|^{t|V(\tau)|}(R_4t)^{t|E(\tau)|}A^t\]
	and using \cref{lem: var_bound},
	\begin{align*}
		\var{f(Y)} &\le 2|E(\tau)|4^{|E(\tau)|}\max_{r \in [|E(\tau) \setminus E(S_0)|]} (\mu_0(f, Y) \mu_r(f, Y)4^rL^r)\\
		&\le 2|E(\tau)|16^{|E(\tau)|} |V(\tau)|^{2|V(\tau)|}A^2
	\end{align*}
	Putting them together, we get
	\begin{align*}
		\EE[(f(Y) - \EE[f(Y)])^t] &\le \max\bigg\{\bigg(\sqrt{2tR_4^{|E(\tau)|}|E(\tau)|16^{|E(\tau)|} |V(\tau)|^{2|V(\tau)|}A^2}\bigg)^t, |V(\tau)|^{t|V(\tau)|}(R_4t)^{t|E(\tau)|}A^t\bigg\}\\
		&\le (C't)^{t|E(\tau)|}|V(\tau)|^{t|V(\tau)|}A^t
	\end{align*}
	for an absolute constant $C' > 0$.
	Finally, $\EE[f(Y)] \le L^{|E(S_0)|} |\cR| \le L^{|E(S_0)|}n^{|V(\tau) \setminus S_0|} \le A$ which gives
	\begin{align*}
		\EE[f(Y)^t] \le 2^t(\EE[(f(Y) - \EE[f(Y)])^t] + \EE[f(Y)]^t)
		&\le 2^t((C't)^{t|E(\tau)|}|V(\tau)|^{t|V(\tau)|}A^t + A^t)\\
		&\le (Ct)^{t|E(\tau)|}|V(\tau)|^{t|V(\tau)|}A^t
	\end{align*}
	for an absolute constant $C > 0$.
\end{proof}